\documentclass[12pt]{article} 

\usepackage{amssymb, amsmath, amscd}
\usepackage{graphicx}
\usepackage[latin1]{inputenc}
\usepackage{epsfig}
\usepackage{color}
\usepackage{amssymb, amsmath, amscd}
\usepackage{graphicx}
\usepackage{subfig}
\usepackage[latin1]{inputenc}
 \usepackage{tikz}
\usepackage{tqft}
\usepackage{mathrsfs}
\usepackage{tikz-cd}
\usepackage{float}

\numberwithin{equation}{section}
\newtheorem{lemma}{Lemma}
\newtheorem{definition}{Definition}
\newtheorem{theorem}{Theorem}
\newtheorem{example}{Example}
\newtheorem{proposition}{Proposition}
\newtheorem{proof}{Proof}

\usepackage{pdflscape}
\usepackage{afterpage}
\usepackage{capt-of}
\newenvironment{system}[1][]%
	{\begin{eqnarray} #1 \left\{ \begin{array}{lll}}%
	{\end{array} \right. \end{eqnarray}}
    
\begin{document}

% MACROS these R. Dridi
% ---------------------
% caligraphique
\newcommand{\calA}{{\cal A}}
\newcommand{\calB}{{\cal B}}
\newcommand{\calF}{{\cal F}}
\newcommand{\calG}{{\cal G}}
\newcommand{\calR}{{\cal R}}
\def\D{\mathcal{D}}
\def\L{\mathcal{L}}
\def\S{\mathcal{S}}
\def\I{\mathcal{I}}
\def\V{\mathcal{V}}
\def\E{\mathcal{E}}
\def\H{$(\mathcal{H})$}
% Ensemble de nombres
\newcommand{\N}{{\mathbb N}}
\newcommand{\Z}{{\mathbb Z}}
\newcommand{\Q}{{\mathbb Q}}
\newcommand{\R}{{\mathbb R}}
\newcommand{\C}{{\mathbb C}}
\newcommand{\K}{{\mathbb K}}

\newcommand{\Frac}[2]{\displaystyle \frac{#1}{#2}}
\newcommand{\Sum}[2]{\displaystyle{\sum_{#1}^{#2}}}
\newcommand{\Prod}[2]{\displaystyle{\prod_{#1}^{#2}}}
\newcommand{\Int}[2]{\displaystyle{\int_{#1}^{#2}}}
\newcommand{\Lim}[1]{\displaystyle{\lim_{#1}\ }}

\newcommand{\DF}[4]{\Frac{\partial^2 f_{#1 #2}}{\partial u_{#3} \partial u_{#4}}}

\newcommand{\tab}{\hspace*{\fill}}
\newcommand{\bs}{{\backslash}}
\newcommand{\eps}{{\varepsilon}}
\newcommand{\into}{{\;\rightarrow\;}}
\def\Hat{\widehat}
\def\Bar{\overline}
\def\vect{\vec}
\def\fbar{{\bar f}}
\def\xbar{{\bar \x}}

\def\Maple{{\sc Maple}}
\def\RG{{\sc Rosenfeld-Gr\"obner}}
%\newcommand{\iff}{{\;\longleftrightarrow\;}}

%%%%%%  mathrm
\def \rank{\mathrm{rang}}
\def \ord{\mathrm{ord}}
\def \dim{\mathrm{dim}}
\def \aut{\mathrm{Aut}^{\mathrm{loc}} }
\def \Id{\mathrm{Id}}
\def \Hom{\mathrm{Hom}}
\def \span{\mathrm{span}}
\def \codim{\mathrm{codim}}
\def \Ker{\mathrm{ker}}
\def \coKer{\mathrm{coker}}
\def \diff{\mathrm{Diff}^{\mathrm{loc}} }
\def \diffg{\mathrm{Diff} }
\def \rp{\mathrm{rp}}
\def \card{\mathrm{card}}
\def \leader{\mathrm{ld}}
\def \Esc{\mathrm{Esc}}

\def\J{\mathrm{J}}
\def\x{\mathrm{x}}
\def\a{\mathrm{a}}
\def\d{\mathrm{d}}

%%%%% Groupes
\def\GL{\mathrm{GL}}
\def\det{\mathrm{det}}
\def\SL{\mathrm{SL}}
\def\PSL{\mathrm{PSL}}
\def\PGL{\mathrm{PGL}}
\def\O{\mathrm{O}}

%%% Algebres de Lie
\def\gl{\mathfrak{gl}}
\def\g{\mathfrak{g}}
\def\h{\mathfrak{h}}

\def\Bar{\widehat}

\newenvironment{Abstract}%
        {\begin{small}\centerline{\bf Abstract} \begin{quote}}
        {\end{quote} \end{small}}

\newenvironment{iitemize}
    {\begin{list}%
      {--}%
      {\setlength{\parsep}{2.5pt}
       \setlength{\itemsep}{2.5pt}
       \setlength{\topsep}{0pt}
       \setlength{\partopsep}{0pt}
      }
    }
    {\end{list}}

%\newcounter{menumi}
%\newenvironment{menumerate}%
%    {\begin{list}%
%      {(\roman{menumi})}%
%      {\usecounter{menumi}%
%       \setlength{\parsep}{0pt}%
%       \setlength{\itemsep}{0pt}%
%       \setlength{\topsep}{0pt}%
%       \setlength{\partopsep}{5pt}%
%       \setlength{\labelwidth}{10pt}%
%       \setlength{\labelsep}{5pt}%
%       \setlength{\leftmargin}{\labelwidth + \labelsep}%
%      }%
%    }%
%    {\end{list}}

%\newcounter{menumi}
\newenvironment{menumerate}{%
    \renewcommand{\theenumi}{\roman{enumi}}% 
    \renewcommand{\labelenumi}{\rm(\theenumi)}%
    \begin{enumerate}} {\end{enumerate}}

\newcommand{\afaire}{$$\vdots$$ \begin{center} {\sc a faire ...} \end{center} $$\vdots$$ }

\newcommand{\pref}[1]{(\ref{#1})}

% systemes d'equations
%\newenvironment{system}[1][]%
%{\begin{eqnarray} #1 \left\{ \begin{array}{lll}}%
%{\end{array} \right. \end{eqnarray}}

\newenvironment{meqnarray}%
	{\begin{eqnarray}  \begin{array}{rcl}}%
	{\end{array}  \end{eqnarray}}

\newenvironment{Pmatrix}
        {$ \left( \!\! \begin{array}{rr} } 
        {\end{array} \!\! \right) $}

%
%\newtheorem{definition}{DÈfinition}
%\newtheorem{theorem}{ThÈor\`eme}
%\newtheorem{proposition}{Proposition}
%\newtheorem{lemma}{Lemme}
%\newtheorem{corollary}{Corollaire}
%\newtheorem{example}{Exemple}
%\newtheorem{remark}{Remarque}
%\newtheorem{proof}{DÈmonstration}[chapter]

%%%%% Abbreviations
\def\MGMG{M\times G \times \overline M \times \overline G}
\def\RM{\calR(M)}
\def\cRM{\calR^*(M)}
\def\GM{G \times M}
\def\MG{M \times G}

\def\lbar{\overline}
\def\S{\mathcal{S}}
\def\Bar{\widehat}

\newcommand{\fleche}[3]{#1 \stackrel{#2}\longrightarrow #3}
\def\ssi{si et seulement si\ }

%%%% Algos
%\newcommand{\algf}{\sffamily}
%\newcommand{\FUNCTION}{{\algf fonction}}
%\newcommand{\BEGIN}{{\algf debut}}
%\newcommand{\END}{{\algf fin}}

%\newcommand{\IF}{{\algf si}\ \ \ \=\+}
%\newcommand{\THEN}{{\algf alors}\ }
%\newcommand{\ELSE}{\=\-\kill{\algf sino}\=\+{\algf n}}
%\newcommand{\ELIF}{{\algf sinon}\=\+{\algf si}\ }
%\newcommand{\FI}{\=\-\kill{\algf fin si}\ }
%\newcommand{\FOR}{{\algf pour}\ \=\+}
%\newcommand{\ENDFOR}{\=\-\kill{\algf fin pour}\ }
%\newcommand{\FROM}{{\algf de}\ }
%\newcommand{\TO}{{\algf \a`a}\ }
%\newcommand{\WHILE}{{\algf tq}\ \ \ \=\+}
%\newcommand{\ENDWHILE}{\=\-\kill{\algf fin tq}}
%\newcommand{\DO}{{\algf faire}\ }
%\newcommand{\OD}{{\algf od}}
%\newcommand{\RETURN}{{\algf retourner}}
%\newcommand{\INDENTER}{{\algf si} \=\+\kill}

\newcommand{\algf}{\sffamily}
\newcommand{\BEGIN}{{\algf begin}}
\newcommand{\END}{{\algf end}}
\newcommand{\IF}{{\algf if}}
\newcommand{\THEN}{{\algf then}}
\newcommand{\ELSE}{{\algf else}}
\newcommand{\ELIF}{{\algf elif}}
\newcommand{\FI}{{\algf fi}}
\newcommand{\WHILE}{{\algf while}}
\newcommand{\FOR}{{\algf for}}
\newcommand{\DO}{{\algf do}}
\newcommand{\OD}{{\algf od}}
\newcommand{\RETURN}{{\algf return}}
\newcommand{\PROCEDURE}{{\algf procedure}}
\newcommand{\FUNCTION}{{\algf function}}
\newcommand{\INDENTER}{{\algf si} \=\+\kill}

%% OpÈrateurs

\newcommand{\initial}{\mathop{\mathsf{init}}}
\newcommand{\separant}{\mathop{\mathsf{sep}}}
\newcommand{\rem}{\mathop{\mathsf{rem}}}
\newcommand{\quo}{\mathop{\mathsf{quo}}}
\newcommand{\pquo}{\mathop{\mathsf{pquo}}}
\newcommand{\lcoeff}{\mathop{\mathsf{lcoeff}}}
\newcommand{\mvar}{\mathop{\mathsf{mvar}}}
\newcommand{\ld}{\mathop{\mathrm{ld}}}
\newcommand{\prem}{\mathop{\mathsf{prem}}}
\newcommand{\remp}{\mathrel{\mathsf{partial\_rem}}}
\newcommand{\remf}{\mathrel{\mathsf{full\_rem}}}
\renewcommand{\gcd}{\mathop{\mathrm{gcd}}}
\renewcommand{\deg}{\mathop{\mathrm{deg}}}
\newcommand{\pairs}{\mathop{\mathrm{pairs}}}
\newcommand{\dd}{\mathrm{d}}
\newcommand{\ideal}[1]{(#1)}
\newcommand{\cont}{\mathop{\mathrm{cont}}}
\newcommand{\pp}{\mathop{\mathrm{pp}}}
\newcommand{\pgcd}{\mathop{\mathrm{pgcd}}}
\newcommand{\ppmc}{\mathop{\mathrm{ppcm}}}
\newcommand{\init}{\mathop{\mathrm{initial}}}
\newcommand{\NF}{\mathop{\mathrm{NF}}}
\newcommand{\rang}{\mathop{\mathrm{rang}}}
\newcommand{\Fr}{\mathop{\mathrm{Fr}}}

\newcommand{\trdeg}{\mathop{\mathrm{tr~deg}}}
\newcommand{\kk}{{\mathrm k}}

\title{Homological Description of the Quantum Adiabatic  Evolution
With a View Toward Quantum Computations}

\author{Raouf Dridi\footnote{rdridi@andrew.cmu.edu}, \, Hedayat Alghassi\footnote{halghassi@cmu.edu} ,\, Sridhar Tayur\footnote{stayur@cmu.edu} \\ {\small {\sc Quantum Computing Group}}\\
{\small {\sc Tepper School of Business} }  \\
{\small {\sc Carnegie Mellon University,  Pittsburgh, PA 15213}}\\
%{\small $\{$rdridi, halghass, stayur$\}$@andrew.cmu.edu}
}

\date{\today}

\maketitle

{\small 

\begin{abstract}
%Global AQE = Morse homology. Speedup = GaussBonnet
We import the tools of  Morse theory to study quantum adiabatic
evolution, the core mechanism in adiabatic quantum computations (AQC).
AQC is computationally equivalent to the (pre-eminent paradigm) of the Gate
model but less error-prone,  so it is ideally suitable to practically tackle a
large number of important applications. 
AQC remains, however, poorly understood theoretically and its mathematical
underpinnings are yet to be satisfactorily identified.  Through Morse theory, we
bring a novel perspective that we expect will open the door for using
such mathematics in the realm of quantum computations, providing a secure
foundation for AQC. Here we show that the singular homology of a certain cobordism, which
we construct from the given Hamiltonian,
defines the adiabatic evolution.  Our result is based on E. Witten's  
construction for Morse homology that was derived in the very different context
of supersymmetric quantum mechanics. We investigate how such topological
description, in conjunction with   Gau\ss-Bonnet theorem and curvature based reformulation of   Morse lemma,  can be an obstruction to any computational advantage in AQC. {We also explore Conley theory, for the sake of completeness, in advance of any known practical Hamiltonian of interest. We conclude with the instructive case of the ferromagnetic $p-$spin where we show that changing its first order quantum transition (QPT) into a second order QPT, by adding non-stoquastic couplings, amounts to homotopically deform the initial surface accompanied with birth of pairs of critical points. Their number  reaches its maximum when the system is fully non-stoquastic. In  parallel, the total Gaussian curvature gets  redistributed (by the Gau\ss--Bonnet theorem) around the new neighbouring critical points, which weakens the severity of the QPT. }
\end{abstract}

%\newpage
%{\color{blue}
%Adding the following:  
%\begin{itemize}
%    \item[1] {\bf Conley theory}: Consider following Hamiltonian 
%    $$
%    H(s,u)= s(u H_z + \alpha (1-u) H_{xx}) + (1-s) H_x.
%    $$
%    The Hamiltonian is stoquastic iff $\alpha\leq0$. 
%    \begin{itemize}
%       \item It has a monkey saddle. 
%        \item This monkey saddle splits into two saddle points. Speedup.
%    \end{itemize}
%    \item {\bf Guass Bonnet}: state as a proposition (in a dedicated section) the fact that two Hamiltonians with same topology don't necessarily have same speed -- Gauss Bonnet thm.
%    \item {\bf Choices of Morse function}: magnetization, partition function. 
%\end{itemize}
%}

\medskip
\medskip
\medskip

\tableofcontents

\section{Introduction}

% what is a critical point

% define the set critical (f)

% def of Morse function goes here
% ... 
%If the boundary $\partial M=M_0\cup M_1$ is not empty then the function ${f: M\rightarrow [0, 1]}$ is Morse if in addition to the non degeneracy condition above, we have that  1) $f$ has no critical points on $\partial M$ and 2) $f(M_0)=0$ and $f(M_1)=1$ which simply means that $f$  takes constant values on the boundaries.

Quantum algorithms running on quantum computers promise to solve computational problems that are intractable for classical computers.  
 A salient illustration of this computational supremacy is Shor's algorithm~\cite{Shor:1997:PAP:264393.264406},  which solves prime factorization in polynomial time whereas the best classical algorithm takes an exponential time. %As a matter of fact, Shor's algorithm pioneered the great effervescence in quantum computing we witness today. 
Grover's algorithm~\cite{Grover:1996:FQM:237814.237866},  which searches a marked item in large unsorted datasets, is another example that comes with quadratic speed-up over classical counterparts.   These two examples were instrumental in the subsequent effervescence around quantum computing. 
 
~~\\
Today, quantum computing research is  dichotomized essentially around  two paradigms:  the gate model \cite{Nielsen:2011:QCQ:1972505} and adiabatic quantum computing (AQC) \cite{Farhi472}. 
While the  gate  model possesses  robust mathematical foundations,  AQC lacks such  foundations and a deep understanding of its power.
This is unfortunate, because AQC  is not only computationally equivalent to the  gate model \cite{1366223, PhysRevLett.99.070502, Kitaev:2002:CQC:863284}, but also  less error-prone \cite{PhysRevA.65.012322, PhysRevA.74.052322}  and much easier to use for a large number of important applications~(binary optimization problems--an AQC processor is, in fact, an optimizer).    

~~\\
The present paper fills that lacuna -- and lays the required foundation -- by offering a completely novel mathematical depiction of AQC based on beautiful mathematics: Morse theory. Our topological investigation unearths hidden mathematical structures underlying AQC's core mechanism: quantum adiabatic evolution.  The arsenal of tools that comes with such mathematics--such as {\it gradient flows, Morse homology and  Gau\ss-Bonnet theorem}--are weapons we deploy  to quantify essential aspects of the adiabatic evolution. 

\subsection{What is a Morse Function? }
A {\it Morse function} is a function whose critical points are non degenerate. Consider, for instance, the real-valued function 
\begin{equation}
	f(s, \lambda) = \lambda^2 - s^2.
\end{equation}
Its graph is  the saddle surface depicted in 
Figure \ref{saddleSurface}. A {\it critical point} of $f$ is a point
$p=(s, \lambda)$, where the gradient of $f$ vanishes; that is, 
a point at which $\partial_sf(p)= \partial_\lambda f(p)=0$, which yields, the point $p=(0,0)$ in this example. This critical point is {\it non degenerate} because the determinant of  the Hessian of $f$ at $p$ is not zero. In fact, the Hessian of $f$ is
given by the matrix
\begin{equation}
\begin{pmatrix}
\partial_{ss}f(s, \lambda) &\partial_{s\lambda}f(s, \lambda)\\
\partial_{s\lambda}f(s, \lambda) &\partial_{\lambda\lambda}f(s, \lambda)
\end{pmatrix}
=
\begin{pmatrix}
2 &0\\
0 &-2
\end{pmatrix}.
\end{equation}
It has two non-zero eigenvalues of opposite signs, in which case the critical point $p$ is called a {\it saddle point}. A non degenerate critical point with strictly negative eigenvalues (of the Hessian) is called a {\it maximum} (or a {\it source}). For instance, $(0,0)$
is a maximum for the function $f(s, \lambda) = -\lambda^2 - s^2$ (the graph of which is the reversed bowl in Figure \ref{hdecomposition} (c)). Similarly, a non degenerate critical point with strictly positive eigenvalues is called a {\it minimum} (or a {\it sink}) (e.g., Figure \ref{hdecomposition} (a)).

    \begin{figure}[h]%
    \centering
    \subfloat  
    {{\includegraphics[width=5cm]{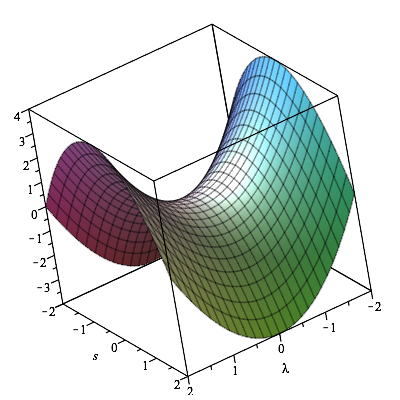} }}%
   
    \caption{{\footnotesize Graph of the function $f(s, \lambda)=\lambda^2-s^2$.  
	}}%
    \label{saddleSurface}%
\end{figure}

\subsection{A Primer on Morse Theory}
Morse theory stems from the observation that important topological properties of
smooth manifolds can be obtained from the critical points of Morse  functions on them. 
To understand this, 
let $M$ be a manifold with--possibly empty--boundaries $M_0$ and $M_1$; a Morse theorist would say
$M$ is a {\it cobordism} from $M_0$ to $M_1$ and denote it by $(M, \partial M)$,  where
$\partial M= M_0\sqcup M_1$.  {{Figure \ref{fig:example}}} gives 
 two examples of cobordisms with two different topologies. The one on the left (tea pot-like) has an empty boundary (i.e., is a closed surface), and the one on the right has a non empty boundary ($M_0$
 is the lower circular boundary and $M_1$ is the disjoint union of the two upper circular boundaries). 
 %A point $p$ in $M$ is said to be a {\it critical point} of the real-valued function $f\in C^\infty(M)$ if the differential map 
%$d f(p)$
%is identically zero which amounts to equating the partial derivatives of $f$ at $p$ to zero, once a coordinates system is chosen.
%The critical point $p$ is said to be {\it non degenerate} if the determinant of the Hessian of $f$ at $p$ is not zero. If all the critical points of the function $f$
%are non degenerate, then $f$ is said to be a {\it Morse function. }
In  {Figure \ref{fig:example}} (left), the height function is Morse and has four non degenerate critical points: a {minimum}, a {saddle point}, and two { maxima}. On the second surface, {Figure \ref{fig:example}} (right), the height function is also Morse, only with one critical point: a saddle point.

~~\\
M. Morse's key observation is that, with the knowledge of critical points,  it is possible to reverse engineer the original topology. The tool for that is the powerful {\it handlebodies decomposition} procedure, which we use repeatedly in this paper.  {Figure \ref{hdecomposition}} provides the dictionary between the critical points and the handles that one can use to recover the cobordism on which the Morse function is defined. For instance, the surface on Figure \ref{fig:example} (left) is recovered by glueing a 0-handle, 1-handle, and  two 2-handles, corresponding to the minimum, the saddle point,   and the two maxima, respectively.

  \begin{figure}[h]
    \centering
    \subfloat  %[]
    {{\includegraphics[width=3cm]{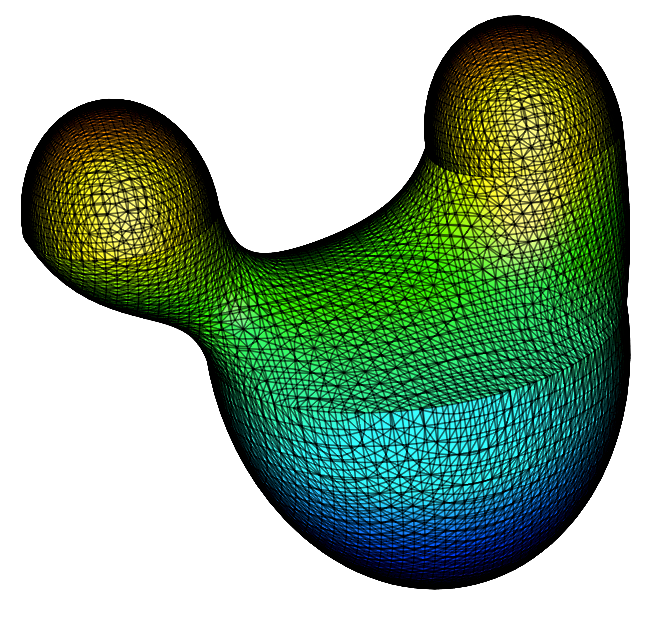} }}%
    \qquad
    \subfloat %[]
     {{\includegraphics[width=2cm]{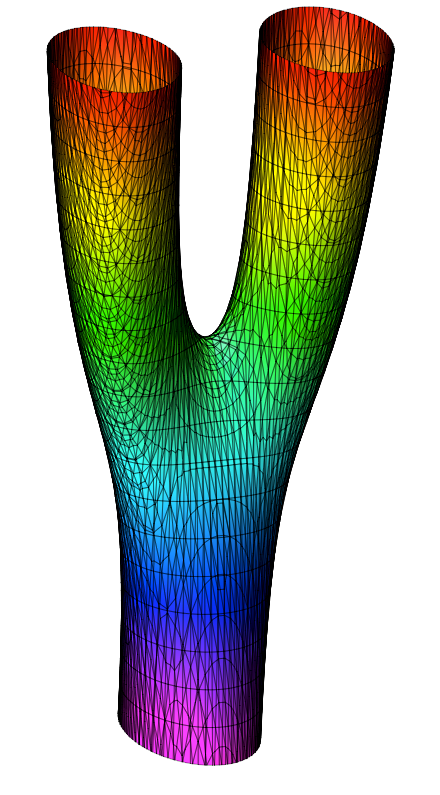} }}%
    %\quad
     %\includegraphics[width=3cm]{pics/curvaturepantsEdited.png}
    \caption{{\footnotesize  (Left) Deformed Sphere. The height function is Morse and has four critical points.   (Right) The so-called pair of pants has a different topology (with Euler characteristic $\chi= \# minima -\# saddles + \# maxima=-1$, compared to 2 for the sphere). The height function is also Morse but has only one critical point. The pair of pants is the cobordism that is assigned to
    Grover's search Hamiltonian. 
        % Also depicted the Dupin indicatrix  at the point $p$. The 
    	%Dupin indicatrix is  defined by ${\mathrm{II}_p(w)=k_1(p) \xi^2 +  k_2(p) \eta^2=\pm 1}$ with 
	 %here $k_1(p)=(1-2^n)$ and  $k_2(p)=1$.  - check homeomorphic type...% mathworld.wolfram.com/HomeomorphicType.html}
	 %(Right) Gauss map on the principal curvatures directions. Degree of n {is~{-1}}, counter-clockwise rotation. (Bottom)
	 %Curvature estimate for spec gap.
	}}%
    \label{fig:example}%
\end{figure}

    \begin{figure}[h]%
    \centering
    \subfloat  [0-handle]
    {{\includegraphics[width=2cm]{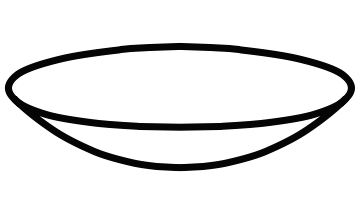} }}%
    \quad
    \subfloat [1-handle]
     {{\includegraphics[width=4cm]{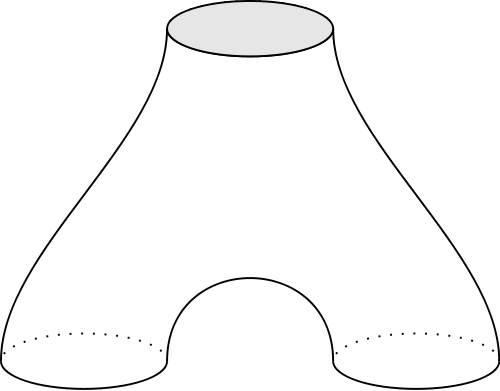} }}%
        \quad
    \subfloat [2-handle]
     {{\includegraphics[width=2cm]{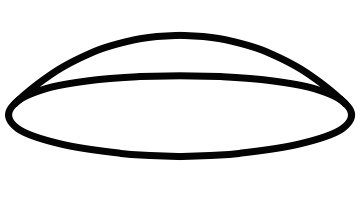} }}%
    %\quad
     %\includegraphics[width=3cm]{../pics/curvaturepantsEdited.png}
    \caption{{\footnotesize Handles.  
	}}%
    \label{hdecomposition}%
\end{figure}

% ~~\\
%Homologically,  the relation between the topology of~$M$ and the critical points is expressed 
%with the so-called Morse inequalities where the Betti number $\beta_i$ of the singular homology of~$M$ gives an upper bound for the number of critical points of index~$i$
%(a good reference for the algebraic topological notions employed here is \cite{Hurtubise}). In the case of the closed surface of {Figure  1}, one has for instance $\beta_1=0$ (since the surface is a continuous deformation of the 2 dimensional sphere  and thus contains no hole) which indicates that the saddle point can be cancelled with a better choice of Morse function (e.g., by projecting onto the horizontal axis). The \textit{pair of pants}
% however has $\beta_1=2$, thus indicating that the saddle point there can not be cancelled.  
 
 ~~\\
In the language of homology,  the relation between the topology of~$M$ and the critical points is expressed 
with the so-called Morse inequalities, where the Betti number $\beta_i$ of the singular homology of~$M$ gives an upper bound for the number of critical points of index~$i$.
(A good reference for the algebraic topological notions employed here -- such as homology, Betti numbers, and Euler characteristic -- is the excellent book \cite{Hurtubise} which is also our main reference for Morse theory). In the example of the closed surface of {Figure  \ref{fig:example}}, one has $\beta_1=0$ (since the surface is a continuous deformation of the 2 dimensional sphere  and thus contains no hole), which indicates that the saddle point can be canceled with a better choice of Morse function (e.g., by projecting onto the horizontal axis). The \textit{pair of pants},
 however, has $\beta_1=2$, indicating that the number of saddle points can not exceed 2   (which can be attained by bending towards the top the dangling leg).    

~~\\
 The full power of Morse theory was unleashed by R. Thom and S. Smale  (the latter with his work on the {Poincar\'e} conjecture) in the sixties and  subsequently, in the eighties, by E. Witten. 
 S. Smale introduced the dynamical system point of view  which has played a central role since then; in particular, in  Witten's explicit construction of the Morse complex \cite{witten1982}.   To the Morse function $f$ we assign the (downward) gradient
 flow given by 
 \begin{equation}\label{gfintro}
 x'(\tau)=- (\nabla f)\circ x(\tau)
 \end{equation}
 for smooth curves $x(\tau)\in (M, \partial M)$.  The  trajectories (also called  instantons ) define a complex whose homology turns out to be isomorphic to the singular homology of the given manifold. This is the apex result of Morse theory, known as the Morse homology theorem, which plays a central role in this paper. 
 
 ~~\\
{The consideration of degenerate critical points started with 
C. Conley in~\cite{0821816888}.} His observation  is that many of the constructions can be carried out with a less local approach. In particular, the behaviour of trajectories around a critical point is indicative of its nature. Therefore, it suffices to consider isolated neighbourhoods around critical points. Within such neighbourhood, continuous deformations can split the degenerate critical into a set of non degenerate points without changing the dynamics around the isolated neighbourhood. The interest is then on how trajectories connect the  isolated neighbourhoods of the function $f$.{We briefly explore this towards the end of the paper. }   
%The second major development is Floer homology, introduced by A. Floer, in his proof of the Arnold conjecture in % symplectic geometry. This beautiful subject is, however, beyond the scope of the current paper.  
%The chains are the critical points of the Morse function graded by their index. 

\subsection{Goal of The Paper and Summary of The Results}
In the  present paper we 
import the  tools of Morse theory to study quantum adiabatic evolution, where we  describe topologically, the {\it adiabatic solutions}    of the Schr\"odinger equation (with $\hbar$ is set to 1): 
\begin{equation}\label{shrodinger}
	\mathrm{i}  \Frac{\partial }{\partial t} |\varphi (t)\rangle = H(t) |\varphi (t)\rangle.
\end{equation}
Adiabatic solutions, central to AQC, are solutions obtained after passing to the slow regime ($s= \varepsilon t$ with $0<\varepsilon <1$) at the singular limit $\varepsilon\mapsto 0$. Here $H(t)$ is 
  some Hamiltonian of interest  operating on the Hilbert space~$\mathcal H ={\mathbb C^2}^{\otimes n}$ -- we assume usual smoothness assumptions on $H(t)$ and $H'(t)$ and their spectral projections.  
  
 ~~\\
 Our topological expedition starts with the consideration
of the function
\begin{equation}
f(s, \lambda) = det(H(s)-\lambda I), 
\end{equation}
that is, the characteristic polynomial of the Hamiltonian $H(s)$ in the slow regime. Our results are ramifications of this choice. In particular,  if  $f$ is Morse,  then  Morse theory is applied without difficulty,  resulting in a number of direct topological reformulations of the quantum mechanical objects. Indeed, 
we automatically inherit  a gradient flow given by the Morse function $f$. This gradient flow  is defined on  a cobordism $(M, \partial M)$ that is also  straightforwardly inherited  from the critical points of $f$ using the handlebodies decomposition.  The eigenenergies (energies of the adiabatic solutions),  which are defined by $f(s, \lambda)=0$,  manifest themselves  as level sets of  the gradient flow. 
As a matter of fact, in the light of the Morse homology theorem, these level sets  are subjected to the singular homology of the cobordism
$(M, \partial M)$: given the set of the critical points of $f$, the Morse homology theorem yields a ``procedure"
that connects the   critical points  consistently with the singular homology of the cobordism $(M, \partial M)$.  These connections are the flow trajectories, and by orthogonality, we obtain the level sets,  in particular, the eigenenergies. We extend these results
to degenerate critical points in the case of the so-called $k-$fold saddle points. Conley theory can be applied, as well, without major difficulties, leading to same conclusions as in the non degenerate case. 

~~\\ 
Our voyage doesn't end here. We start our descent from the topological global description to a  {\em local} description of the quantum adiabatic evolution around the critical points of $f$.  Essential to our local description is the curvature at the critical points: two Hamiltonians having the same topology are not necessarily ``computationally equivalent" i.e., have different speedups. We use differential geometry to obtain the quantitative behaviour of the eigenenergies around the critical points. In fact, given the homology of the cobordism $(M, \partial M)$,  Gau\ss-Bonnet theorem distributes the Gaussian curvature of~$(M, \partial M)$ consistently with this homology --but not necessarily in the same way for topologically equivalent Hamiltonians. By re-expressing  Morse lemma in terms of the principal curvatures,  one can  obtain the delay factor at the given critical point.   We explain the role of the shape operator to describe the adiabatic evolution -- including the quantitative behaviour of the eigenenergies. In fact, since the shape operator is Hermitian (self-adjoint), one might think of this procedure as {\it dimensionality reduction} of the original Hamiltonian $H(s)$. 

 \subsection{An Illustrative Example: Quantum Search} 
As an introductory example, let us consider the adiabatic Hamiltonian for the search problem~\cite{cerf, vazirani}:
\begin{equation}\label{HGroverIntro}
	H(s)  = (1-s) H_{initial} + sH_{final},
\end{equation}
where $s=t/T$, with $T>0$, and 
\begin{eqnarray}
H_{initial}  &=&  1 -  |\hat 0\rangle \langle \hat 0|,  \\[3mm]
H_{final}  &=&  \sum_{z\in \{0, 1\}^n - \{ u\} }   |z\rangle \langle z| = 1 -  |u\rangle \langle u|. 
\end{eqnarray}
 As usual, the notations $ \{|z\rangle\}_{z\in \{0, 1\}^n}$ and  $\{|\hat z\rangle\}_{z\in \{0, 1\}^n}$ stand for the computational and Hadamard bases, respectively. The state $ |u\rangle$ is the sought item (the unsorted database being  the computational basis $ \{|z\rangle\}_{z\in \{0, 1\}^n}$). The search problem can be put into 
    the two-dimensional subspace spanned by the two  states   $|v\rangle := |\hat 0\rangle-  1/\sqrt{N} |u\rangle$ and $|u\rangle$, with $N=2^n$. In this orthogonal basis, the restricted Hamiltonian $H(s)$ takes the form 
\begin{equation}
H(s)=  
\begin{pmatrix} 
{\frac { \left( 1-s \right)  \left( N-1
 \right) }{N}}&{\frac { \left( -1+s \right)  \left( N-1 \right) }{{N}^
{3/2}}}\\ \noalign{\medskip}{\frac {s-1}{\sqrt {N}}}&{\frac {1-s+sN}{N}
} 
\end{pmatrix}.
\end{equation}
%    
%\begin{equation}
%H(s)= \left[ \begin {array}{cc} {\frac { \left( 1-s \right)  \left( N-1
% \right) }{N}}&{\frac { \left( -1+s \right)  \left( N-1 \right) }{{N}^
%{3/2}}}\\ \noalign{\medskip}{\frac {s-1}{\sqrt {N}}}&{\frac {1-s+sN}{N}
%}\end {array} \right] 
%\end{equation}
The characteristic polynomial of this $2\times 2$  matrix is 
\begin{equation}\label{MorseGrover}
	f(s, \lambda) = {\frac { \left( N-1 \right) }{{N}}} (s-s^2) + \lambda^2-\lambda,
\end{equation}
and has one critical point $p$ obtained by equating its partial derivatives to zero. This critical point
is non degenerate because the eigenvalues 
\begin{equation}
{k_1(p) = -2\,{\frac { \left( N-1 \right) }{{N}}}}\, \mbox{ and } {k_2(p) =2}    
\end{equation}
 of the Hessian of $f$ are non zero, and because $k_1(p) k_2(p)<0$, the critical point is a saddle point. Now, the graph of the function $f$ (see Figure \ref{saddleSurface}) comes with a Gaussian curvature $K(s, \lambda)=(f_{ss}f_{\lambda \lambda}-f_{\lambda}^2)/(1+f_s^2+f_\lambda^2)^2$.  Gau\ss-Bonnet theorem forces this  curvature~to distribute itself on the surface consistently with this topology (consistent with Euler characteristic~{-1}). %And since the surface has only one  critical point,  the curvature is concentrated around it. 
In fact, the curvature is ``dumped" at the critical point $p$:  
\begin{equation}
	\int_{(M, \partial M)} K(s, \lambda) d\sigma \sim \int_{V(p)} K d\sigma  = -2\pi +   O(1/N),
\end{equation}
where $V(p)$ is an arbitrary small neighbourhood  around the saddle point~$p$, and  independent of $n$. Explicitly, 
we have, $K(p)=k_1(p)k_2(p) = -4(1-\frac{1}{N})$ and the two quantities $k_1(p)$ and $k_2(p)$ are the two principal curvatures of the saddle surface at $p$.  
If  we intersect the  surface with planes horizontal to the tangent plane  $T_pM$, in particular the plane $f(s, \lambda)=0$, we obtain two hyperbola called {\it Dupin indicatrix}. The  radius $g(s)$ of this indicatrix (that is, the distance  between the two hyperbola), which is also the energy difference, 
is constrained by the amount of curvature at $p$. Indeed, we have  
\begin{equation}
	g(s) = 2 {\frac {\sqrt {-k_{{2}}(p) \left( 2{  f(p)}+k_{{1}}(p){(s-\frac{1}{2})}^{2} \right) }}{k
_{{2}}(p)}},
\end{equation}
and from which we infer:  
\begin{equation}
	\int_0^1 \frac{ds}{g(s)^2}  =   
    {\frac {-k_{{2}}(p)}{\sqrt {k_{{1}}(p) {{  f(p)}}}}}
    \arctan \left(  {\frac {\sqrt {k_{{1}}(p)} }{\sqrt {{
 8 f(p)}}}} \right)  
 =
 \frac{\pi}{2} \,\sqrt {N}-1+O \left( {\frac {1}{\sqrt {N}}} \right),
\end{equation}
which is the total time needed to tunnel through the saddle point $p$ without destroying the adiabaticity. 
Our approach thus provides a new derivation for quadratic speedup.

\subsection{Outline of The Paper}
The paper is organized as follows. Section 2 reviews the adiabatic theorem. The version we review is valid not only for both discrete and continuous spectrum, but also  in the presence of the eigenvalues crossing. In section 3, we summarize properties of gradient flows and review  Morse lemma. Section 4 connects the quantum adiabatic evolution to gradient flows. Section 5 has two parts: (1) global description of the quantum  adiabatic evolution based on Morse homology, and~{(2) (local)} description of the Gaussian curvature around the critical points.  Gau\ss-Bonnet theorem bridges  the two parts. Section  6 generalizes the findings of Section 5 to the degenerate case -- specifically, to the $k$-fold saddle points. {Section 7 is dedicated to the ferromagnetic $p-$spin model. This instructive model exhibit a first order quantum phase transition that can be changed into
a second order transition by the addition of non-stoquastic couplings. We describe this as a homotopy deformation accompanied by death and birth of critical points. This number is maximized when the system is fully  non-stoquastic.}
%We conclude in Section 8 with {some open questions.}

~~\\

~~\\

\begin{center}
\begin{table}[h]
\begin{tabular}{c|c}
{\bf Quantum adiabatic evolution} & {\bf Morse theory }\\
\hline \\
Hamiltonian $H(s)$ & Morse function $f(s, \lambda)=det(H(s)-\lambda I)$\\
\hline\\
Eigenenergies & Level sets of the gradient flow:\\
& $x'(\tau)=- (\nabla f)\circ x(\tau)$\\
\hline\\
Eigenenergies & are orthogonal to the boundary maps\\
& of the Morse homology\\
%\hline\\
%Shrinking of the spectral gap & Saddle points of $f$\\
%\hline\\
%Rate of shrinking of the spectral gap & Rate of distribution of the Gaussian curvature around saddle points\\
\hline\\
Spectral gap & Radius of Dupin Indicatrix\\
\hline\\
Dimensionality reduction of $H(s)$ & The shape operator\\
\hline\\
Degeneracy & $k-$fold saddle points\\
\hline
\end{tabular}
\caption{Correspondence between the quantum  adiabatic evolution and Morse theory.}
\end{table}
\end{center}

%{\color{blue} Morse lemma revisited + delay factor + story}
%
%{\color{blue}
%\begin{table}[h]
%\centering
%\begin{tabular}{|c|c|}
%\hline\\
%Adiabatic evolution & Geometry \\
%\hline \hline\\
%Adiabatic flow & A gradient flow $x(\tau)' = -\nabla f(x(\tau))$ on a cobordism $(M, \partial M)$\\
%\hline\\
%Spectrum  & level set of the gradient flow $x(\tau)' = -\nabla f(x(\tau))$\\
%\hline\\
%... &  ...\\
%\hline\\
%... &  ...\\
%\hline 
%\end{tabular}
%\caption{Correspondence between optimization problems and adiabatic quantum computations.}
%\end{table}
%}
%
%
%{\color{blue}
%\section{Notations}
%\begin{itemize} 
%	\item The Hilbert space here is ${\mathbb C^2}^{\otimes n}$.
%	\item $H(t) \in  {\bf SelfAdjointOperators}( {\mathbb C^2}^{\otimes n})$ a time dependent Hamiltonian.
%	\item $(M, \partial M)$ a cobordim i.e.,  {\color{blue} a finite compact
%	Riemannian manifold} $M$  with  boundary $\partial M$. 
%	\item $m= \mathrm{dim}(M)$.
%	\item $X$  a smooth vector field on  $M$. 
%	\item $x(\tau) \subset M$ a curve in $(M, \partial M)$  (might intersect with the boundary).
%	\item $critical(f) $ set of critical points of $f$.
%	\item $\gamma(p)$ index of the critical point $p \in critical(f)$.
%\end{itemize}
%}
%\newpage
\section{The Adiabatic Theorem}
The  adiabatic theorem is an existence result of solutions of the Schr\"odinger equation that goes back to the early days of 
quantum mechanics.
 It describes both 
the solutions and the regime (i.e., conditions) in which such  solutions exist. Physically, this regime is characterized by ``slowly" varying the time dependent Hamiltonian.
Mathematically,  this is done by considering the Schr\"odinger equation 
\begin{equation}\label{shrodinger2}
	\mathrm{i} \Frac{\partial }{\partial t} |\varphi (t)\rangle  = H(t) |\varphi (t)\rangle,
\end{equation} 
with the new (slow) time $s= \varepsilon t$ with $0<\varepsilon <1$ which yields:
\begin{equation}\label{shrodingerWithS}
	\mathrm{i} \varepsilon \Frac{\partial }{\partial s} |{\varphi}_\varepsilon (s)\rangle  
		=  H(s) |{\varphi}_\varepsilon(s)\rangle.
\end{equation} 
 The adiabatic theorem describes  solutions at the singular limit $\varepsilon\rightarrow 0$. 
 In the literature,   
there is no single adiabatic theorem, and different theorems focus on  different assumptions on the Hamiltonian (\cite{Born1928, 
Kato, MR1027662, Avron1999} to cite a few). They do, however,  share  the following
structure \cite{Avron1999}: 
\begin{theorem}
Let $P(s)$ be a spectral projection of $H(s)$. In the singular limit~${\varepsilon\rightarrow 0}$, the solution
$|{\varphi}_\varepsilon (s)\rangle$ of (\ref{shrodingerWithS}) with the initial condition 
$|{\varphi}_\varepsilon (0)\rangle \in Range\left(P(0)\right)$ is subject to 
\begin{equation}\label{adiabaticThm}
	dist \left( |{\varphi}_\varepsilon (s)\rangle, \,  Range\left( P(s)\right)\right) \leq O(\varepsilon^\gamma) 
\end{equation}
for an appropriate value of $\gamma\geq 0$ depending on the assumptions made about  $H(s)$. 
\end{theorem}
This formulation is particularly interesting not only  because it is valid for discrete and continuous spectrum, but also because it continues to hold in the event of an eigenvalues crossing.  Recall that eigenvalues crossing refer to when $H$ has two eigenvalues $\lambda^i$ and~$\lambda^j$ 
 that are isolated from the rest of the spectrum and  
  equal to each other at some time $s$. The order of  such a crossing is by definition the order of the zero of 
$\lambda^i(s)-\lambda^j(s)$.

~~\\
In their original statement of the adiabatic theorem, 
 Fock and Born \cite{Born1928}  restricted themselves to Hamiltonians with simple discrete spectrum. They have showed that  if the eigenvalues crossing doesn't happen, then  $\gamma\geq 1$; that is, the gap is a smooth function of $\varepsilon$. 
 To see this, suppose
the  spectral projection 
$P(s)$ is given by $|{\varphi}^i (s)\rangle \langle {\varphi}^i (s)|$ and consider 
the {\it ansatz} 
\begin{equation}
	|{\varphi}_\varepsilon (s)\rangle = \mathrm{exp} \left\{\phi(s, \varepsilon) \right\} |{\varphi}^i (s)\rangle
\end{equation}
with  the complex-valued function $\phi$   required to satisfy $\phi(0, \varepsilon)=0$. 
The substitution in the equation~(\ref{shrodingerWithS}) yields
$$   
	\mathrm{i} \varepsilon \partial_s |{\varphi}^i (s)\rangle +
	\mathrm{i} \varepsilon \partial_s \phi(s, \varepsilon) |{\varphi}^i (s)\rangle = \lambda^i(s)   
	|{\varphi}^i (s)\rangle.
$$
Assuming $\varepsilon$ is close to 0, the first term of the left hand side can be neglected. Integrating gives
$$
	\phi(s, \varepsilon) =  -\frac{\mathrm{i}}{\varepsilon}\int_{0}^{s} \lambda^i(s) ds + O(\varepsilon). 
$$	
Thus,
\begin{equation}\label{adiabaticSol0}
 |\varphi_\varepsilon(s)\rangle =
	  \mathrm{exp}\left\{ -\frac{\mathrm{i}}{\varepsilon}\int_{0}^{s} \lambda^i(s) ds\right\} |\varphi^i (s)\rangle + O(\varepsilon). 
\end{equation}
To see where the crossing comes in, we need to expand this   approximation to the first order. 
Plugging 
\begin{equation}\label{plugging2}
 |\varphi_\varepsilon(s)\rangle=\mathrm{exp}\left\{ -\frac{\mathrm{i}}{\varepsilon}\int_{0}^{s} \lambda^i(s) ds \right\} \left( |\varphi^i (s)\rangle + \varepsilon |\zeta(s)\rangle + O(\varepsilon^2) \right)
 \end{equation}
in the Schr\"odinger
and solving for the function $ |\zeta(s)\rangle$  yields  
\begin{eqnarray}\label{zeta}
	 |\zeta(s)\rangle = \sum_{j\neq i} 
	 		  \left(\frac{A(0)}{(\lambda^j(0) -\lambda^i(0))^2} - \mathrm{exp}\left\{\frac{\mathrm{i}}{\varepsilon}  (\lambda^j(s) - \lambda^i(s))\right\} \frac{A(s)}{(\lambda^j(s) -\lambda^i(s))^2} \right)|\varphi^j (s)\rangle 
\end{eqnarray}
with 
\begin{equation}
A(s)=  \mathrm{i} \langle \varphi^i(s) | \frac{d}{d s}H(s) | \varphi^j(s)\rangle.
\end{equation}
 The term of interest is the denominator~${(\lambda^j(s) -\lambda^i(s))^2}$, which is not defined   when eigenvalues crossing occurs -- which is also true for the higher terms in $ O(\varepsilon^2)$. If the eigenvalues crossing is excluded,  the expression \pref{plugging2},  with $ |\zeta(s)\rangle$ given by~(\ref{zeta}), is a well defined solution.

 \subsection{Eigenvalues Crossings and AQC}
% The adiabatic theorem as presented above, continues to be valid in the event of eigenvalue crossings occurs.
%In AQC, where the $H(s)$ is of the form
%$H(s)  = \alpha(s) H_{initial} +\beta(s) H_{final}$,  the  non-crossing rule below forbids eigenvalue crossings from happenning if the time $s$ is the only parameter to control the evolution.
As a matter of fact,  in their original paper,
Born and Fock also studied  eigenvalues crossings and obtained  expansions similar
to  (\ref{adiabaticSol0}), with $O(\varepsilon)$ replaced with $O(\varepsilon^{1/(m+1)})$, with $m$ the order of the crossing 
\begin{equation}\label{adiabaticSolm}
	|\varphi_\varepsilon(s)\rangle =
	  \mathrm{exp}\left\{ -\frac{\mathrm{i}}{\varepsilon}\int_{0}^{s} \lambda^i(s) ds\right\} |\varphi^i (s)\rangle  + O(\varepsilon^{1/(m+1)}).
\end{equation}
The source of 
this fractional power is that near the crossing point ${|\lambda^j(s) -\lambda^i(s)| = s^m}$ (See also  \cite{MR1027662}).  That being said, in the context of current implementations of AQC,
 where the time dependent Hamiltonian $H(s)$ has a particular form (as in \pref{HGroverIntro}),  the eigenvalues crossing doesn't occur if  the time $s$ is the only parameter that drives the evolution. 
 This is a corollary of the so-called {\it non-crossing rule} (or avoided crossing rule) which is another important early result, due to von Neumann and Wigner 
\cite{vonNeumann1993}.  The proof can be found in \cite{Suzuki2013} or in the original paper~\cite{vonNeumann1993}. 

\begin{theorem}[Non-crossing rule]
Suppose the Hamiltonian $H$ depends on a number of independent real parameters, and
suppose that $H$ has $r$ eigenvalues with multiplicity $m_i$ for $i=1,\cdots, r$. 
Assume that no observable commutes with the Hamiltonian. In this case,  
 the number of the free real parameters to fix $H$ is $2^{2n}+r -\sum_{i-1}^rm_i^2$. 
\end{theorem}
 Suppose the Hamiltonian is given by 
 \begin{equation}\label{HGroverIntro}
	H(s)  = \alpha(s) H_{initial} + \beta(s)H_{final}
\end{equation}
where $H_{initial} $  and $H_{final}$ are non commuting observables, and 
$\alpha(s)$ and $\beta(s)$ are two functions that control the evolution (e.g., $\alpha(s)=1-s$ and $\beta(s)=s$). Suppose $s_0 < s_1$ are  two different times such that for all $s\in [s_0, s_1]$ both
functions $\alpha(s)$ and $\beta(s)$ are not zero. The theorem above excludes the  scenario where $H(s_0)$ has a simple spectrum  whilst  $H(s_1)$ has a doubly degenerated eigenenergy.  To see this,  it suffices to count the number of the free real parameters in each case (that is, for $H(s_0)$ and $H(s_1)$) to see that  three parameters  in $H(s)$ are  to be fixed to obtain $H(s_1)$ (fixing $s$ to $s_1$ in addition to two more parameters). Since $\alpha(s)$ and $\beta(s)$ depend only on $s$, such a scenario can not happen. However, when more parameters (at least two more) are involved, the eigenvalues crossing is plausible. 
%{\color{red} real vs complex Hermitian! see von Neumann paper! }

%When one of the functions $\alpha(s)$ or $\beta(s)$ is zero the Hamiltonian is essentially an observable.  

%A direct corollary of this theorem is when we consider two Hamiltonians $H_{initial}$ and $H_{final}$, where
%the first has a simple spectrum (that is, $r=2^n$ and $m_i=1$ for all $i=1,\cdots, r$ giving $2^{2n}$ free parameters),  and the
%second hamiltonian $H_{final}$ has one doubly degenerated eigenenergy (that is, $r=d-1$, $m_1=2$ 
%and $m_i$ for $i=2,\cdots, r$ which give
%  $2^{2n}-3$ free parameters).  
%
%
%In the context of AQC ... Grover ... Saddle points

\section{Morse Lemma and Gradient Flows}
  Let  $M$ be a cobordism between $M_0$ to $M_1$ and 
 $f:(M, \partial M)\rightarrow [0,1]$   a smooth function.
  
 \begin{definition}
 Let $p\in M$ be a non degenerate critical point of $f$.  
The  number  $\gamma(p)$  of negative eigenvalues of the Hessian of $f$ at $p$  is called the Morse index of~$f$ at~$p$. 
\end{definition}
 The Morse index is an intrinsic invariant for the homeomorphism group of~$M$ (i.e.,  one-to-one mappings on $M$ continuous in both directions).
 
%{\color{blue}
\begin{example}
In the 2-dimensional case where the Hessian is a 2 by 2 matrix,  
 the index  is necessarily in $\{0, 1, 2\}$. 
\end{example}
%}
Around its non degenerate critical points, the function $f$ takes a very simple form:  
\begin{lemma}[Morse lemma]
	Let $p$ be a non degenerate critical point of a smooth function $f: M\rightarrow [0, 1]$.  There exists a neighbourhood $U$
	of $p$ and a diffeomorphism~${h: (U, p)\rightarrow (\mathbb R^m, 0)}$ such that
	$$
		f\circ h(x_1, \cdots, x_m) = f(c)  - \sum_{1\leq i\leq \gamma(p)} x^2_i + \sum_{ \gamma(p)+1\leq i\leq m} x^2_i,
	$$
	with $m= \mathrm{dim}(M)$.
\end{lemma}
 An immediate corollary of  Morse lemma is
that the non degenerate critical points of a Morse function are isolated. Another direct consequence is that, in  
  the 2-dimensional case, a critical point with index values 0, 1 and 2, respectively, corresponds to a local minimum, a saddle point,  and a local maximum.  

\begin{definition}
The neighbourhood $U$ and its image $h(U)$ in   Morse lemma are  called, respectively, the {\it manifold chart} and  {\it Morse chart}.
\end{definition}
 
 \begin{example}[Quantum search]
  Continuing (from the Introduction) with Grover's search Hamiltonian with
  $f(s, \lambda) = (2^{-n}-1) s^2 +  \lambda^2$,    where we have translated the saddle point to the origin. Figure  \ref{charts} (a) gives the Morse chart around the origin. Figure \ref{charts} (b) gives the manifold chart. 
    \begin{figure}
    \centering
    \subfloat  [Morse chart for Grover's search Hamiltonian $H_{\sf Grover}(t)$. Also depicted, the negative gradient field in blue.  Orthogonal to the integral lines, in red, are  level sets  in black. ]
    {{\includegraphics[width=4.8cm]{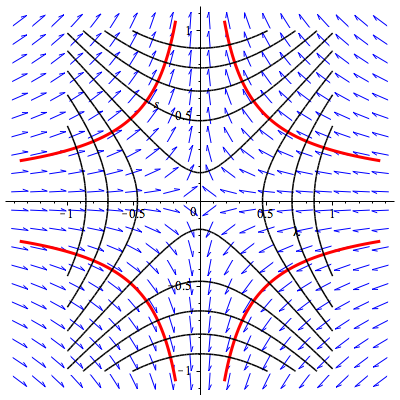} }}%
    \qquad
        \centering
    \subfloat  [ Manifold chart for Grover's search. Orthogonal to the integral lines, in red, are  level sets  in black. ]
    {{\includegraphics[width=5.5cm]{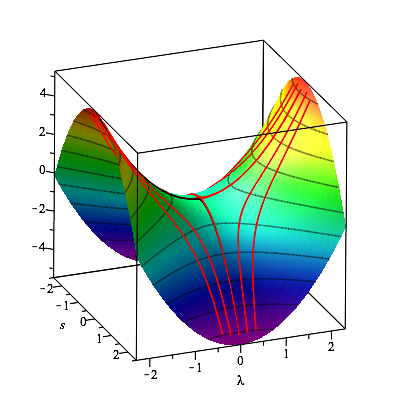} }}%
    \caption{ Morse and manifold charts for Grover's search Hamiltonian.}
    \label{charts}
  \end{figure}  
 \end{example}  
 
 ~~\\
 {Remarkably, the different local descriptions of $f$ can be glued together. The fact that these local depictions can be glued together is the incarnation of the globality of Morse theory. In the language of {\it functors}, this can be presented  as follows: Let $ \mathcal O_M$ denote the category of standard open sets of $M$
sorted by inclusion and $C^\infty(M)$ denote the category of real-valued smooth functions sorted by restrictions. 
\begin{proposition}
The functor $\mathcal F_M:  \mathcal O_M \rightarrow C^\infty(M)$, defined by the Morse lemma for each open set, is a sheaf. The Morse function $f$ is the unique global section of~$\mathcal F_M$. 
\end{proposition}
For readers familiar with this language, the proof is almost evident, particularly when  we take into account that critical points are isolated. }
 
 ~~\\
 Let $X$ be a smooth vector field on  the cobordism $M$; that is, 
\begin{equation}
	X= \sum_1^m \xi_i (x) \partial_{x_i}
\end{equation}
where $x=(x_1, \cdots, x_m)$ coordinates $M$ and $\xi_i$ are smooth functions  on $M$. For $x^0\in M,$
 consider the initial value problem (IVP) for smooth curves $x: \mathbb R\rightarrow M$ given by
\begin{equation}\label{IVP}
 	\frac{d{x}(\tau)}{d \tau} = X(x(\tau)), \quad x(0)=x^0.
\end{equation}	
Solutions $x(\tau)$  are called {\it trajectories} (also  {\it flow lines}  or {\it instantons}); their properties are summarized in the following proposition (\cite{Hurtubise}):

\begin{proposition}
 The following  assertions are  true: 
\begin{enumerate}
 \setlength\itemsep{0.0em}
\item   the set of all trajectories  of the IVP problem \pref{IVP}  covers $M$;
\item if  $M$ is closed, then the trajectory $x(\tau)$ is defined for all $\tau\in \mathbb R$;  
\item the limits  $\mathrm {lim}_{\tau\rightarrow - \infty} x(\tau)$ and $\mathrm {lim}_{\tau\rightarrow + \infty} x(\tau)$ 
are two critical points of $f$ (this gives a procedure to locate critical points: starting anywhere on the manifold $M$ and following the flow);  
\item trajectories can escape to the boundary;   
\item trajectories intersect only at critical points.  
\end{enumerate}
\end{proposition}
%\begin{proof}
%	See \cite{Hurtubise}.
%\end{proof}
The flow generated by $X$ is the smooth map 
${\Phi_{-}: \mathbb R\times M\rightarrow M}$, which sends the pair~$(\tau, x^0)$ to the point $\Phi_{\tau}(x^0)=x(\tau)$, where $x(\tau)$
is the solution of the initial value problem \pref{IVP}. The set $\{\Phi_\tau\}_\tau$
is a one parameter group of diffeomorphisms on $M$ with $\Phi_{\tau+\tau'} = \Phi_\tau\Phi_{\tau'}$. The orbit $\mathcal O(x^0)$
through $x^0$ is defined by the set~{$\{ \Phi_\tau(x^0)|\, \tau\in \mathbb R\}$}.  
\begin{definition}
The flow associated with 
$
X = -\nabla f = -\sum f_{x_i}(x) \partial_{x_i}$ is called the  gradient flow. 
\end{definition}
%{\color{blue}
For gradient flows, we have two more properties in addition to the ones listed in the above proposition. First,  if $x(\tau)$
is a trajectory, then $\frac{d}{d\tau}f(x(\tau))  = -|\nabla f(x(\tau))|^2$, which shows that $f$ is decreasing along trajectories (and zero only at the critical points). Second, at any regular point of  the gradient $\nabla f$, trajectories  are orthogonal to the level sets
 $f(x)=\mathrm{constant}$. This second property is manifestly important in the next section. 
%} 

% {\color{blue} existence of Morse functions - page 47/61}
 
 \section{Gradient Flows for Quantum Adiabatic Evolution}
Let $\{\varphi_\varepsilon (s) \}$ be the collection of the adiabatic solutions one gets for the different initial conditions
 (these solutions are given by (\ref{adiabaticSol0}) or  (\ref{adiabaticSolm}), depending whether eigenvalues crossings occur or not). 
 Our aim now is to assign a gradient flow to this collection of solutions. In fact, the gradient flow we construct below is ``orthogonal" to this adiabatic flow, in the sense that the adiabatic solutions are level sets of this gradient flow. As mentioned before, we consider 
 the smooth  function $f (s, \lambda) = det(H(s)-\lambda I)$, which itself gives rise to  the gradient flow 
\begin{system}\label{gradientFlowEnergies}
	\Frac{d}{d \tau} \lambda(\tau) &=& - \partial_\lambda f (\lambda(\tau), s(\tau)),\\[3mm]
	\Frac{d}{d \tau} s(\tau) &=& - \partial_s f (\lambda(\tau), s(\tau)).
\end{system}%
The function $f$ is defined on a open set $U\subset \mathbb R^2$, which we assume   big enough to include all the critical point of $f$ (which itself is assumed to not possess critical points at the infinities).
 Later on, the subset $U$ will be replaced with a cobordism in $\mathbb R^3$
constructed from the critical points of $f$, using the  handle decomposition procedure. In this case, $U$ will be play the role of a Morse chart.
In addition to this picture, we have
a similar picture for a Morse chart $V := (U\cap s-axis)  \times \mathrm{span}_{\mathbb R} \left(\{\varphi_\varepsilon \} \right)  \subset \mathbb R\times \mathcal H$  corresponding to the level sets given by the curves $\{\varphi_\varepsilon (s) \}$.
 Indeed, we can ``lift" the function $f$ to $V$ as follows:

\begin{center}
\begin{tikzcd}
V  \arrow[rd, "g= f\circ h"] \arrow[d, "h"]  & \\
 U  \arrow[r, "f"] & \mathbb R
\end{tikzcd}
\end{center}
with $h$ being the linear map: 
\begin{equation}\label{h}
	h\left(s, \varphi \right) = \left(s, \langle \varphi (s) |H(s)| \varphi (s)\rangle\right). 
\end{equation}
%(defined on the basis states, from which one can extend $g$ to the the span of the eigenstates). 
%We have then 
%$$
%	g\left(s, a_i \varphi_i  + a_i \varphi_i  \right) = f \left(s, a_i \tilde \lambda_i(s) + a_i \tilde \lambda_i(s)\right).
%$$ 
%for any linear combination $a_i \varphi_i(s) + a_i \varphi_i(s)$.  
In particular, $g(s, |\varphi_\varepsilon \rangle)=0 + O(\varepsilon^{2\gamma})$ for the adiabatic solution $|\varphi_\varepsilon(s)\rangle$; this is because  $\langle \varphi_\varepsilon(s) |H(s)| \varphi_\varepsilon(s)\rangle = \lambda^i(s)  + O(\varepsilon^{2\gamma})$. This means that the function~$\tilde f$ is constant on the adiabatic solutions; that is,  the latter  are
level sets  for the gradient flow  
\begin{system}\label{gradientFlowStates}
	\Frac{d}{d \tau} \varphi^{(j)}(\tau) &=& - \partial_{\varphi^{(j)}} g(\varphi(\tau), s(\tau))\\[3mm]
	\Frac{d}{d \tau} s(\tau) &=& - \partial_s g(\varphi(\tau), s(\tau)). 
\end{system}%
Here, $\varphi^{(j)} $ the $j$-th component of the vector $|\varphi \rangle\in \mathrm{span}_{\mathbb R} \left(\{\varphi_\varepsilon \} \right)$.   

\begin{proposition}
	The two flows (\ref{gradientFlowEnergies}) and (\ref{gradientFlowStates}) are topologically equivalent. 
\end{proposition}
 Recall that two autonomous systems of ordinary differential equations $x' = f(x),\, x\in M$ and $y'=g(y), \, y\in M'$ are said to be topologically equivalent~\cite{MR947141} if there exists a homeomorphism   
 $h:M\rightarrow M', \, y=h(x)$, which maps solutions of the first system to the second, preserving the direction of time. 
\begin{proof} Since the direction of time is preserved (i.e., same independent variable $\tau$ in both systems) we  need only to prove that $h$ given by (\ref{h}) is a homeomorphism. For that we prove the two systems
have the same critical points, which implies that $h$ is invertible outside this set of critical points. 
The fact that
$\nabla g(x) = \nabla f(h(x)) \nabla h(x)$ for $x\in V$ implies that 
%critical points of the dynamical system (\ref{gradientFlowEnergies})
%are also critical points for the system (\ref{gradientFlowStates}).  
the set of critical points of  the dynamical system (\ref{gradientFlowStates}) is 
the union of the set of critical points of (\ref{gradientFlowEnergies}) and   the set of critical points of $h$.  By the implicit function theorem,
critical points of $h$  are where $h$ fails to be injective, i.e., points $(s, \varphi_i(s))$ and $(s, \varphi_i(s))$  with
$\lambda_i(s)=\lambda_j(s)$. But these correspond to critical points (intersections of level sets) of (\ref{gradientFlowEnergies}). Thus,  the two dynamical systems
(\ref{gradientFlowEnergies}) and (\ref{gradientFlowStates}) have the same critical points. 
 \end{proof}

~~\\
In this section we assigned two topologically equivalent gradient flows (\ref{gradientFlowEnergies} and \ref{gradientFlowStates})   to the adiabatic solutions of the Schr\"odinger equation \pref{shrodingerWithS}. The adiabatic solutions (respectively,  their energies) are level curves orthogonal to the integral curves of the gradient flow  \ref{gradientFlowEnergies} (respectively, \ref{gradientFlowStates}).
In the next section, we concentrate on  the implication of this  perspective   when $f$ is Morse.

% {\color{blue} Types of critical points for AQE... (relate to crossing non-crossing) }

\section{The Non-Degenerate Case: Morse Theory}
The goal of this section is to prove Theorem \ref{thm1} below, which relates the adiabatic
flow associated to the Hamiltonian $H(s)$ to the topology of a certain cobordism~$(M, \partial M)$.  Our proof  is a direct  corollary of the different constructions we have built and shall build below.  
We start   with  the construction of the cobordism~$(M, \partial M)$.

\subsection{Handlebodies Decomposition: The Cobordism of The Quantum  Adiabatic Evolution} 
Consider   the gradient flow \pref{gradientFlowEnergies}  where the Morse function $f (s, \lambda) = det(H(s)-\lambda I)$ is
defined on a open compact set $U\subset \mathbb R^2$ (the Morse chart), which we have assumed to be big enough to include all the critical points of $f$ (which themselves are assumed to be non degenerate). We would like to extend the domain of $f$ to  a  cobordism $(M, \partial M)$  using the technique of {\em handlebodies decomposition}s which we  review now. Define 
%Let $f: M\rightarrow [0,1]$ be a Morse function. 
%The compactness of $M$ implies
%that the number of the critical points of $f$ is finite. 
 the sublevel set (also half-space) $M_c$ to be the set of points of $M$ at which $f$
takes values less than or equal to $c$: 
$$
M_{c} := \{x\in M|\, \quad f(x)< c\}
$$
Let $p$ be a non degenerate  critical point of $f$ with index $\gamma(p)$, and let
$D^{i}$ denotes the $i$-dimensional unit disc $\{x \in \mathbb R^i \, \mathrm{ such that} \, |x| \leq 1 \}.$
The ${\gamma(p)} $-handle  is defined to be 
the set
\begin{equation}
D^{\gamma(p)} \times D^{m-\gamma(p)}.  
\end{equation}
Examples of handles are given in Figure \ref{hdecomposition}.
%Then the idea behind  handlebody decomposition is to track the changes in $M_c$ as $c$ varies.  In particular, we have: 
We have the following result  (for proof see for instance \cite{MR1873233,Hurtubise}): 
\begin{proposition}\label{handlebodies} 
Assuming, for sufficiently small $\varepsilon>0$,  the only critical
value in $[c-\varepsilon, c+\varepsilon]$ is $c=f(p)$, then the sublevel set 
$
	M_{c+\varepsilon} 
$
can be obtained from the sublevel set $M_{c-\varepsilon}$ by attaching the ${\gamma(p)} $-handle. 
\end{proposition}
 Now, let us assume that the critical values are such that
$
	f(p_1)< f(p_2) <\cdots <f(p_k).
$
If this is not the case, 
we may  perturb $f$ in such a way that $f(p_i)\neq f(p_j)$ if $i\neq j$ (See for instance Theorem 2.34, Chapter 2 in \cite{MR1873233}).  Proposition  \ref{handlebodies} says that starting from the lowest critical value   $f(p_1)$ one can then reverse engineer the manifold by inductively  attaching  ${\gamma(p_i)} $-handle to the sublevel set $M_{f(p_{i-1})}$.  

~~\\ Returning to our context of adiabatic evolution we obtain the following proposition:
\begin{proposition}\label{hbprop}
	Let $f(s, \lambda)= det\left(H(s) -\lambda I\right)$. If $f$ is Morse, then 
	the spectrum of $H(s)$   corresponds to level sets of the gradient flow 	
	(\ref{gradientFlowEnergies}) defined on a cobordism $(M, \partial M)$, which  obtained, from the critical points of $f$,
	using the handlebodies decomposition procedure.
\end{proposition}
In the next subsection, will explain how the gradient flow of $f$  triangulates the cobordism~$(M, \partial M)$.
 
 \subsection{Morse Homology for Non-Degenerate Quantum Adiabatic Evolution}

 \subsubsection{Transversality and Morse-Smale functions }
 For completeness, we start  by reviewing the notion of Morse-Smale functions. The motivation for this  ``technical"  notion is that not all Morse functions give triangulation of the cobordism $(M, \partial M)$. Such functions need to satisfy an additional requirement: Morse-Smale transversality, which, in reality, is not a major constraint; in general, it suffices to ``tilt'' the Morse function slightly to obtain Morse-Smale transversality -- The height function of a vertical torus is a  typical example. 
In general, 
a {\it Smale-Morse function} is a Morse function~$f:M\rightarrow [0,1]$, where the unstable manifold
$W^u(q):= \{q\} \cup \{x\in M|\, \mathrm{lim}_{\tau\rightarrow -\infty} \Phi_\tau(x)=q \}$ and the stable manifold
$W^s(p) :=\{p\}\cup \{x\in M|\, \mathrm{lim}_{\tau\rightarrow +\infty} \Phi_\tau(x)=p \}$ of $f$ intersect transversally for all
$p$ and $q$ in $critical(f)$.

 %\subsubsection{Orientation}
 %This subsubsection is also included for completeness. The notion of orientation is important in the construction of the Morse complex when the field is $\mathbb Z$.  
 
%~~\\ 
% Let $V$ be  a real vector space of finite dimension. On the set of ordered bases of $V$
% we define the equivalence relation $v\mathcal R w$ if and only if the change of basis
% matrix between $v$ and $w$ has a positive determinant.  A {\it positive orientation} (or, for short,
%  {\it orientation})
% of $V$ is a choice of one of the classes of the equivalence relation $\mathcal R$.
% For instance, the standard orientation on the vector space $\mathbb R^m$ is the orientation defined by it standard basis. 

\subsubsection{The integers $n(p,q)$: Counting flow lines}
Consider two critical points $p$ and $q$ of index $\gamma(p)=k-1$ and $\gamma(q)=k$ respectively, and 
assume that the intersection $W(q, p)=W^u(q)\cap W^s(p)$ is  not empty. Let $x(\tau):\mathbb R\rightarrow M$ be a gradient flow line
from $q$ to $p$:
\begin{equation}
	\frac{d}{d\tau} x(\tau) = - (\nabla f) (x(\tau)), \quad lim_{\tau\rightarrow -\infty} x(\tau) =q,\, lim_{\tau\rightarrow \infty} x(\tau) =	p.
\end{equation}
%At any point $x\in x(\mathbb R) \subset W(q, p)$ we can complete the vector
%  $-\left(\nabla f\right)(x)$ to a positive basis $\left(-\left(\nabla f\right)(x), \tilde B^u_x \right) $ of 
%  $T_xW^u(q)$. If we pick any positive basis $B^s_x$ of   $T_xW^s(p)$ then $(B_x^s, \tilde B_x^u)$ is a basis
%  for $T_xM$. If this basis is a positive orientation for $T_xM,$ then we assign +1 to the flow line $\lambda$.
%  Otherwise we assign -1.
  Now since  $\gamma(q)-\gamma(p)=1,$ the set (moduli space) $\mathcal M(q, p):= W(q, p)/\mathbb R$  of flow lines connecting $q$
  and $p$ is zero dimensional. In fact, it consists of a finite number of elements (Corollary 6.29 \cite{Hurtubise}). 
 % To each flow line  in $\mathcal M(q, p)$
 % we have assigned a number +1 or -1 using the orientations. 
 The integer $n(q, p)\in \mathbb Z_2$ is defined to be the sum mod 2 of these numbers. 
\subsubsection{Morse chain and homology}
 Let $C_k(f)$ be the free abelian group generated by the critical
  points of index~$k$, and define $C_*(f) = \bigoplus_{k} C_k(f).$ 
  The homomorphism $\partial_k: C_k(f) \rightarrow C_{k-1}(f)$ defined by 
\begin{equation}
  	\partial_k(q) = \sum_{p\in {critical}_k(f)}n(q, p)p
\end{equation}
is called the Morse boundary operator, and the pair
$(C_*(f), \partial_*)$ is called the Morse  chain complex of $f$.
  The deep connection between Morse theory and  topology is expressed in the following
  landmark result.
 \begin{proposition} [Morse homology theorem]
 	The pair $(C_*(f), \partial_*)$  is a chain complex, and its   homology  is isomorphic
	to the singular relative homology $H_*(M, \partial M)$.
 \end{proposition}
 \begin{proof}
See for instance  \cite{Hurtubise} Theorem 7.4. 
 \end{proof}

%A nice feature of the Morse-Morse chain complex is 
%that the boundary operator is described geometrically. So, if we can draw the space
%and the Morse-Smale gradient flows, then it is often easy to compute homology using
%Morse Homology Theorem. 

%  \begin{figure}%
%    \centering
%     \subfloat
%     {{\includegraphics[width=7cm]{pics/File_001.jpeg} }}%
%     %\qquad
%     %\subfloat
%     %{{\includegraphics[width=7cm]{pics/File_000.jpeg} }}%
%    \caption{{\footnotesize   
%    More complexes: double Grover.
%    }}%
%    \label{homologies}%
%\end{figure}

 \begin{example}[Deformed sphere -- Figure \ref{homologies}] 
 We have the chain complex
 \begin{equation}
 0\rightarrow  span_{\mathbb{Z}_2} (p_3, p_4)
 \rightarrow span_{\mathbb{Z}_2} (p_2)\rightarrow  span_{\mathbb{Z}_2} (p_1)
 \rightarrow 0,
 \end{equation}
 %$C_0 =\langle p_1\rangle, \, C_1 =\langle p_2\rangle $ and $C_2 =\langle p_3, p_4\rangle$
 with the boundary maps
 $\partial p_4=\partial p_3 =p_2$ and  $\partial p_2 =2p_1\, mod\, 2=0$. There are flow lines connecting the two maxima to the minimum (not drawn below) but these ones are not part of the definition. We obtain  $\beta_0=1,\,  \beta_1=0$, and $\beta_2=1$.
 
  \end{example}	
   \begin{figure}[h]
    \centering
    \subfloat{{\includegraphics[width=11cm]{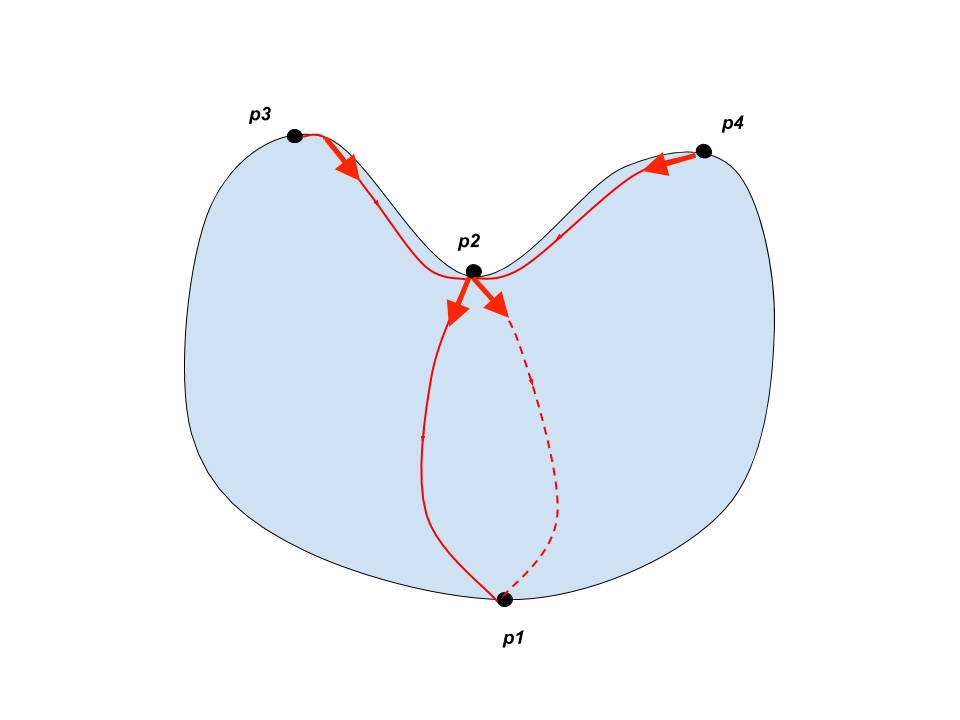} }}%
        
    \caption{{\footnotesize   
    Morse complex for deformed sphere. Flow lines are flowing downhill.
    }}%
    \label{homologies}%
\end{figure}

~~\\
So far, for the sake of 
simplicity of the exposition, we have used $\mathbb{Z}_2$ as the field of coefficients. Morse homology theorem holds, however, with coefficients in any commutative ring  with unity $K$. The theorem is proved for $\mathbb{Z}$, the ring of integers and by a universality argument the proof extends to any commutative ring with unity. When counting the numbers $n(p,q)\in \mathbb{Z}$, a notion of orientation is introduced \cite{Hurtubise}. Now, our main result of this section is the following:

\begin{theorem}\label{thm1}
	Let $f(s, \lambda) = det(H(s)-\lambda I)$, assumed to be Morse, and let $(M, \partial M)$ be  the associated cobordism constructed in Proposition \ref{hbprop}. 
	The adiabatic flow associated with $H(s)$ 
	is uniquely defined by the set ${critical}(f)$ 
	and the singular   homology~$H_*(M,\partial M; K)$.
\end{theorem}
 \begin{proof}
The boundary maps $\partial_k$
 of the Morse complex are uniquely defined by the isomorphism in the Morse homology theorem,
in addition to the set of critical points of $f$. Additionally,   Morse lemma gives $f$ and the gradient flow lines, locally around each critical
point. The number $n(p, q)$ allows us to connect these trajectories consistently. The adiabatic solutions are then level sets of   this gradient flow.
 \end{proof}
%{\color{blue} Needs elaboration on the physics: glueing local energies etc.  Also prepare connection to differential geometry next.}
In the remainder of this section, we start our descent from the  global qualitative description  of the quantum  adiabatic evolution, given by the theorem above, to a local quantitative description around the critical points of $f$.  Essential in our local description is the curvature at the critical points.  The reason for moving to this curvature based local description is the fact that two Hamiltonians giving the same topology are not necessarily ``computationally equivalent" i.e., have the same speedup. Indeed, Gau\ss-Bonnet theorem may distribute the total curvature that comes with this topology differently around the critical points, thus, potentially yielding different ``speeds''. 
%This  well known procedure in differential geometry relies on  Gauss-Bonnet theorem. 

\subsection{Differential Geometry}
We first, review some important constructs in differential geometry of surfaces \cite{doCarmo}.
\subsubsection{Gau\ss\, map and its derivatives}
Let  $S$ be a surface in the Euclidean space $\mathbb R^3$ and $p$ a point in $S$. 
We write  
\begin{equation}
\mathrm{n}:S \rightarrow \mathbb S^1
\end{equation}
 for the
%A central notion in the differential geometry of surfaces is the
 Gau\ss\, map  which  sends $p$ to the normal (unit) vector $\mathrm{n}(p)$.  The rate of change of $\mathrm{n}$ on a neighbourhood of $p$ measures how rapidly the surface $S$ is pulling away from the tangent plane $T_pS$.  Formally, this rate of change is given
by the {\it shape operator} $$d(\mathrm{n})(p): T_pS  \rightarrow   T_{\mathrm{n}(p)}\mathbb S^1,$$ which is a Hermitian operator. The determinant of 
$d(\mathrm{n})(p)$, denoted $K(p)$,   is the {\it Gau\ss\, curvature} at $p$, and its   eigenvalues, denoted by $k_1(p)$ and $k_2(p)$, are the two {\it principal curvatures} at $p$. (These are generalizations, to surfaces, of the notion of the curvature of a curve. In particular, the principal curvatures are the minimum and the maximum of all curvatures of all curves on the surface passing through~$p$).  The corresponding eigenvectors, called principal directions, $e_1$ and $e_2$. They form a convenient orthonormal  basis for~the tangent plane $T_pS$ as will see next.

%\subsubsection{Zooming in: Gau\ss-Bonnet}

\subsubsection{Morse lemma revisited and Dupin indicatrix}
It is  possible  to have a local expression of any smooth function $f(s, \lambda)$ around its non degenerate critical points in terms of the principal curvatures at $p$.  This is obtained by Taylor expanding  $f$ at $p$ 
and then rotating the $s \lambda$-axes to coincide with  the principal directions $e_1$ and $e_2$.
\begin{lemma}[Morse lemma revisited]\label{MorseRevisited} 
	 For each non degenerate critical point $p$ of~$f$, there exists a neighbourhood of $p$  
	such that 
	\begin{equation}\label{approx} 
	 f(\xi, \eta)= f(p) +  \frac{1}{2} \left(k_1 (p) \xi^2 + k_2(p) \eta^2\right) + \mathrm{h.o.t}., 
         \end{equation}
	with $k_1(p)$ and $k_2(p)$ are the principal curvatures of~$S$~at~$p$.
	%(Additionally, ${\sf Hessian}(f)(p) =diag( k_1 , k_2)$ is the associated matrix of the shape operator $dn(p)$).
\end{lemma}
Let $a>0$ be a positive small number. {\it Dupin indicatrix} is  the set of vectors $w$ in $T_pS$ such that 
\begin{equation}\label{dupin}
{\mathrm{II}_p(w)= \pm a},
\end{equation}
with ${\mathrm{II}_p(w)} = \langle d\mathrm{n}(p)(w), w\rangle$  the second fundamental form, which when expanded
with ${w= \xi e_1 + \eta e_2,}$ gives $k_1 (p) \xi^2 + k_2(p) \eta^2$; hence, Dupin indicatrix is a union of conics in $T_pS$.	 Lemma~\ref{MorseRevisited}  says that,  if $p$ is a non degenerate critical point, the intersection with $S$ of a plane parallel to $T_pS$ is, in a first
order approximation, a curve similar to (one of the conics of) the Dupin indicatrix at~$p$.   
  
% Figure xxx depicts the Dupin indicatrix for the two cases: elliptic and hyperbolic.

\subsubsection{Delay factors around saddle points}
Suppose now that the  surface $S$ is the graph of the function  $f(s, \lambda)=det(H(s)-\lambda)$. By Lemma~\ref{MorseRevisited}, if $p$ is a non degenerate critical point of $f$, then  the spectrum  of $H(s)$ -- given by the curve $\{f(s, \lambda)=0\}$ --can be approximated by the Dupin indicatrix  
$\{w=(\xi,\eta)\in \mathbb R^2: \,  {\mathrm{II}_p(w)} =-f(p)\}$,  for some rotation
$(s, \lambda)\mapsto(\xi, \eta)$. If   $p$ is a saddle point    (i.e., $k_1(p) k_2(p)<0$), then the spectrum around $p$ is a set of two hyperbola  similar to those in the search problem  discussed in the Introduction. In particular, the radius of Dupin indicatrix gives     the energy difference, which in turn
%,   is given by   
%\begin{equation}
%	{g(\xi) = 2 {\frac {\sqrt {-{k_2(p)}\, \left( k_1 (p){\xi}^{2}+ f(p) \right) }}{k_2(p)}}}
%\end{equation}
 gives
the  total delay factor.

\subsubsection{Tracking the shape operator}
In light of the above, the   shape operator $d(\mathrm{n})(-)$ emerges as a 
central object   in analyzing the computation advantage of AQC.  Therefore, it is meaningful    to connect  the different locally defined shape operators around the different critical points consistently with the Morse complex; that is, to  restrict the shape operator $d(\mathrm{n})(-)$ to the network $\mathcal N\subset S$ 
consisting of the critical points and the connecting instantons. By doing so, we obtain
%The precise definition of  $\mathcal N$ is as follows.
%First, start with $\mathcal N$ simply as the set of saddle points of the surface $S$ i.e., points $p$ such $K(p)<0$. Now, for each   point, start the corresponding four trajectories and extend them as long as they stay in $S$  and then add (attach) these trajectories to $\mathcal N$. 
a fiber bundle  (a subbundle of the tangent bundle $TS$)
\begin{equation}
	\pi: \cup_{p\in \mathcal N}\left\{w \in T_pS:\,  {\mathrm{II}_p(w) =-f(p)} \right\}\rightarrow \mathcal N
\end{equation}
that captures the adiabatic evolution of Hamiltonian $H(s)$ not only topologically   but also quantitatively around its critical points -- at any saddle point $p\in \mathcal N$ the total time delay can be obtained from the Dupin indicatrix. In fact, at the vicinity of any critical point~$p$, the spectrum of the time dependent Hamiltonian $H(s)$ is completely determined (up to a rotation) by the spectrum of the hermitian  operator $d \mathrm{n}(p)$ acting on the 2-dimensional Hilbert space~$\mathbb R^2$.
%%%%%%%%%

\section{The Degenerate Case: Conley Theory}
We summarize here relevant results from Conley theory--although there is no meaningful Hamiltonian of interest at this time. In Conley's view,  critical points of the function $f$ are represented by the so-called index-pair which we review now -- The following definitions and theorem are taken from \cite{phdthesis}. A compact set $N\subset M$ is said to be an isolating neighborhood of the flow $\Phi_{-}$, if the set
	$
		Inv(N, \Phi_{-})  = \{x\in M|\,  \Phi_\tau(x) \in N \mbox{ for all } \tau \in \mathbb R \}
	$
	is contained in the interior of $N$. An invariant set $S$ is an isolated invariant set
	if there exists an isolating neighborhood $N$ such that $Inv(N, \Phi_{-})=S$. Given
    an isolated invariant set $S$, 
an index pair is  a pair $(N, L)$ where $N$ is an isolating neighborhood, and $L$ is an exit set -- Figure \ref{conleyfig1}. 

  \begin{figure}[h]
    \centering
    \subfloat   
    {{\includegraphics[width=4.8cm]{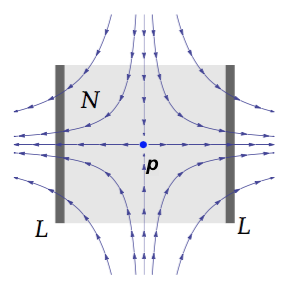} }}%
    \qquad
        \centering
    \subfloat  
    {{\includegraphics[width=5.5cm]{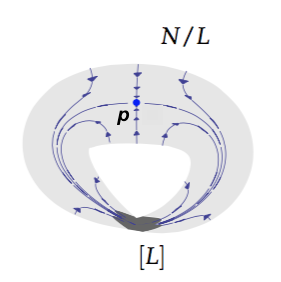} }}%
    \caption{ (Left) An index pair for the saddle point $p$. (Right) The Conley index is the homotopy type of the sphere obtained by collapsing $L$ to a point  (Images by T.O. Rot). }
    \label{conleyfig1}
  \end{figure} 
  
\begin{theorem}[C. Conley]
Let $S$ be an isolated invariant set of the flow $\Phi_{-}$. Then
\begin{itemize}
\item The set $S$ admits an index pair $(N, L)$.
\item The Conley index is the pointed homotopy type $(N\backslash L, [L])$, and is independent of the choice of index pair.
\item For any homotopy of flows $\Phi_{-}^\mu$ (with $\mu\in [0,1]$) such that $N$ is an isolating neighbourhood for all flows $\Phi_{}^\mu$, the Conley index of $Inv(N, \Phi_{-}^\mu)$ remain unchanged for all~$\mu$. 
\end{itemize}
\end{theorem}  
An interesting case is  the case of $k$-fold saddles which are the ``degenerate analogue" of saddle points.  Precisely, a $k$-fold saddle is a critical point at which $f$ is locally given by $\mathscr{R}e \left( \left(x + i y \right) \right)^{k+1}$ with $x,y\in \mathbb R$. An example is depicted in Figure  \ref{conleyfig2}. The  function $f(x, y)=x^3 - 3x y^2$, which can be a characteristic function of some Hamiltonian $H$,  is perturbed into the Morse function $f_\varepsilon(x, y)=x^3 - 3x y^2+ \varepsilon x$. The degenerate critical points bifurcates into two saddle points without affecting the dynamics around, and the results presented in the non degenerate case, can be extended   to this type of degeneracy.

   \begin{figure}[h]
    \centering
    \subfloat   
    {{\includegraphics[width=12cm]{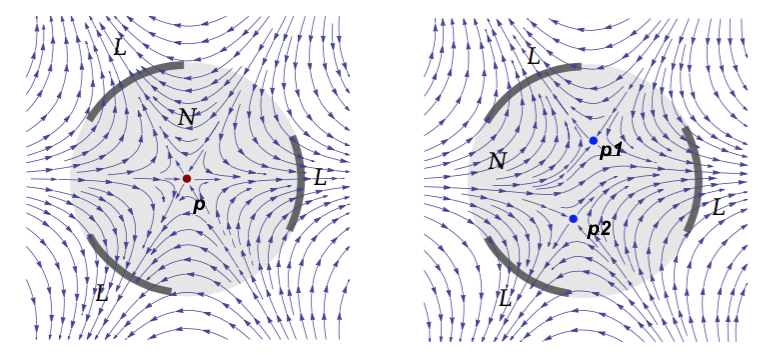} }}%
    \caption{The gradient flow lines of the function $f(x, y)=x^3 - 3x y^2$. Perturbing the critical point $p$ into two saddle points doesn't change the behaviour of the surrounding flow lines  (Images by T.O. Rot).}
    \label{conleyfig2}
  \end{figure}  
 
\section{The case of the ferromagnetic $p-$spin model}
{In this last section,  we pay attention to the  instructive model of the  ferromagnetic $p-$spin. We consider  the time dependent Hamiltonian 
\begin{equation}
	H(s) = -s n\left ( \frac{1}{n} \sum_{i=1}^n \sigma_i^z\right)^p - (1-s) \sum_{i=1}^n \sigma_i^x. 
\end{equation} 
The $p-$th power refers to the matrix power; i.e., composition of linear operators on the Hilbert space ${\mathbb C^2}^{\otimes n}$. This Hamiltonian exhibits a first order quantum phase transition (QPT) for $p\geq 3$;  this can be shown in various ways--for instance,  by taking the thermodynamical limit $n=\infty$, in which case one gets  a classical (continuous) spin of a single particle. Figure \ref{thermo} shows
the continuity and discontinuity of the first derivative of the classical energy $e(p,s)$ when $p=2$ and $p=5$ respectively. 

 \begin{figure}[h]
    \centering
    \subfloat  []
    {{\includegraphics[width=4cm]{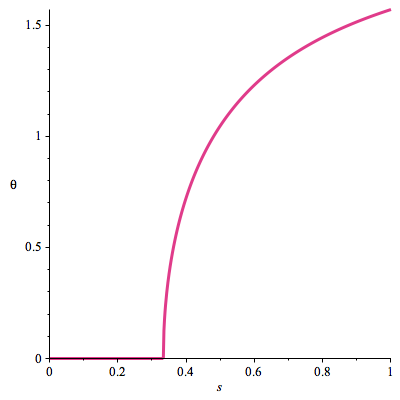} }}% 
        \centering
    \subfloat  [ ]
    {{\includegraphics[width=4cm]{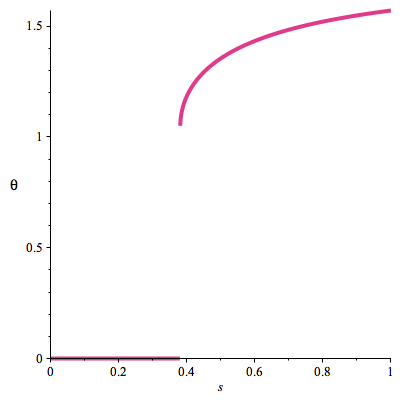} } }%
    \\
      \subfloat  [ ] 
    {{\includegraphics[width=4cm]{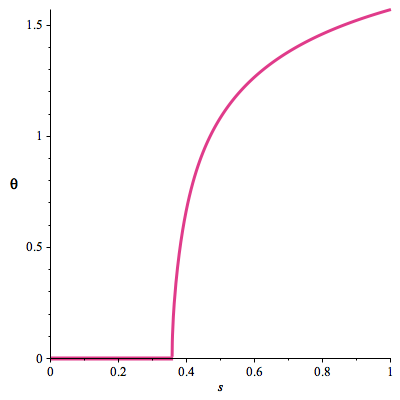} } }%
    \caption{ (a): The minimum of the energy $e(p=2)$ is continuous function of $s$; this minimum is  given by $\theta=\pi -\arccos \left({\frac {-1+s}{2s}} \right)$. The phase transition is second order. (b):
    The minimum of the energy $e(p=5)$ is discontinuous at $s=0.38$. The phase transition is first order. (c) After the addition of the non-stoquastic term with $b =0.1$, the  minimum of the energy is now a continuous function of~$s$. }
    \label{thermo}
  \end{figure}  
 ~\\
For finite $N$, the system is also reduced into a much smaller space because of the commutativity of $H(s)$ with the total spin 
$S$. This reduction can be generalized using the theory of irreducible representations, but here we recall the usual reduction.  Conveniently, the  
 lowest part of the spectrum of the total operator is also the lowest part of the spectrum of $H(s=0)$. And because of the anti-crossing theorem, this is also true for $H(s)$  for $0<s<1$. This lowest part is given by the states with maximum spin $n/2$ i.e.,  states  $|n/2, m\rangle$, for 
 $0\leq m\leq n$,  and defined with 
\begin{eqnarray} 
	S_z |n/2, m\rangle  &=& m |n/2, m\rangle\\ 
	S^2 |n/2, m\rangle  &=& n/2 (n/2+1) |n/2, m\rangle. 
\end{eqnarray}
In addition we have
\begin{equation}
	S_{\pm} |n/2, m\rangle =  \sqrt{n/2(n/2+1) - m(m\pm1)  } \,  |n/2, m \pm 1\rangle
\end{equation}
where $S_{\pm} =S_x \pm S_y$ are the spin raising the lowering operators. All of this  continues to hold with 
the new non-stoquastic Hamiltonian  (\cite{10.3389/fict.2017.00002}): 
\begin{equation}
	H^b(s) =  -s b n\left ( \frac{1}{n} \sum_{i=1}^n \sigma_i^z\right)^p
	+ s(1-b) n  \left( \frac{1}{n} \sum_{i=1}^n \sigma_i^x\right)^k
	 - (1-s) \sum_{i=1}^n \sigma_i^x.  
\end{equation}
%A direct calculation gives the matrix elements: 
%\begin{eqnarray}\nonumber
%	H(s,b)_{m, m} &=&  s \left( -bN \left( 1-2\,{\frac {m-1}{n}} \right) ^{p}+ \left( 1-b
% \right)  \left( 2\,m-1-2\,{\frac { \left( m-1 \right) ^{2}}{N}}
% \right)  \right)
%\\[3mm] \nonumber
%H(s,b)_{m, m+1} &=& H(s,b)_{m+1, m}  = - \left( 1-s \right) \sqrt { \left( N-m+1 \right) m}\\[3mm] \nonumber
%H \left( s,b \right) _{{m,m+2}}&=& H(s,b)_{m+2, m} ={\frac {s \left( 1-b \right) \sqrt {m
% \left( m+1 \right)  \left( N-m+1 \right)  \left( N-m \right) }}{N}}.\nonumber 
%\end{eqnarray}
%The rest of matrix elements are zero.  
The parameter $b$ controls the {\it stoquasticity} of the system \cite{Bravyi:2008:CSL:2011772.2011773}. The relevance of this additional parameter, shown in \cite{10.3389/fict.2017.00002}, is that when $b$ is close to zero (the system is   more and more non-stoquastic),  the first oder QPT changes into a second order QPT (this has important ramifications on the computational power of this particular AQC). For us,  this change, turns out to be, a homotopy:
\begin{equation}\label{h}
h:(M^0, \partial M^0)\times [0, 1] \rightarrow (M^1,\partial M^1),    
\end{equation}
with 
\begin{equation}
h(s, \lambda, b) = f^b(s, \lambda)= det(H^b(s) -\lambda I),    
\end{equation}
 which gives the commutative diagram:
 \begin{center}
\begin{tikzcd}
f^0 \arrow[to=Z, "{h(- ,b)}"] \arrow[to=2-2]
& {} \\
& \qquad  \qquad    H_*(M^b,\partial M^b; \mathbb Z)=H_*(M^0,\partial M^0; \mathbb Z)\\
|[alias=Z]|  f^b  \arrow[to=2-2]
\end{tikzcd}
\end{center}
%where $f^b\in \Gamma \mathcal{F}_{M_b}$ is the unique global section of the functor $\mathcal{F}_{M_b}$. 
Explicitly, figure \ref{dance} shows how critical points are born and dead (by pairs)
as $b$ changes without affecting the Euler characteristic. As the parameter  $b$ decreases towards zero, the number of critical points increases, forcing the total curvature (under the constraint of the Gau\ss--Bonnet theorem) to redistribute itself on a larger set of critical points, which is   behind the change of the first order QPT into a second order \cite{1903.01486}. 
}

\begin{figure}[H]
    \centering
    \subfloat  []
    {{\includegraphics[scale=0.2]{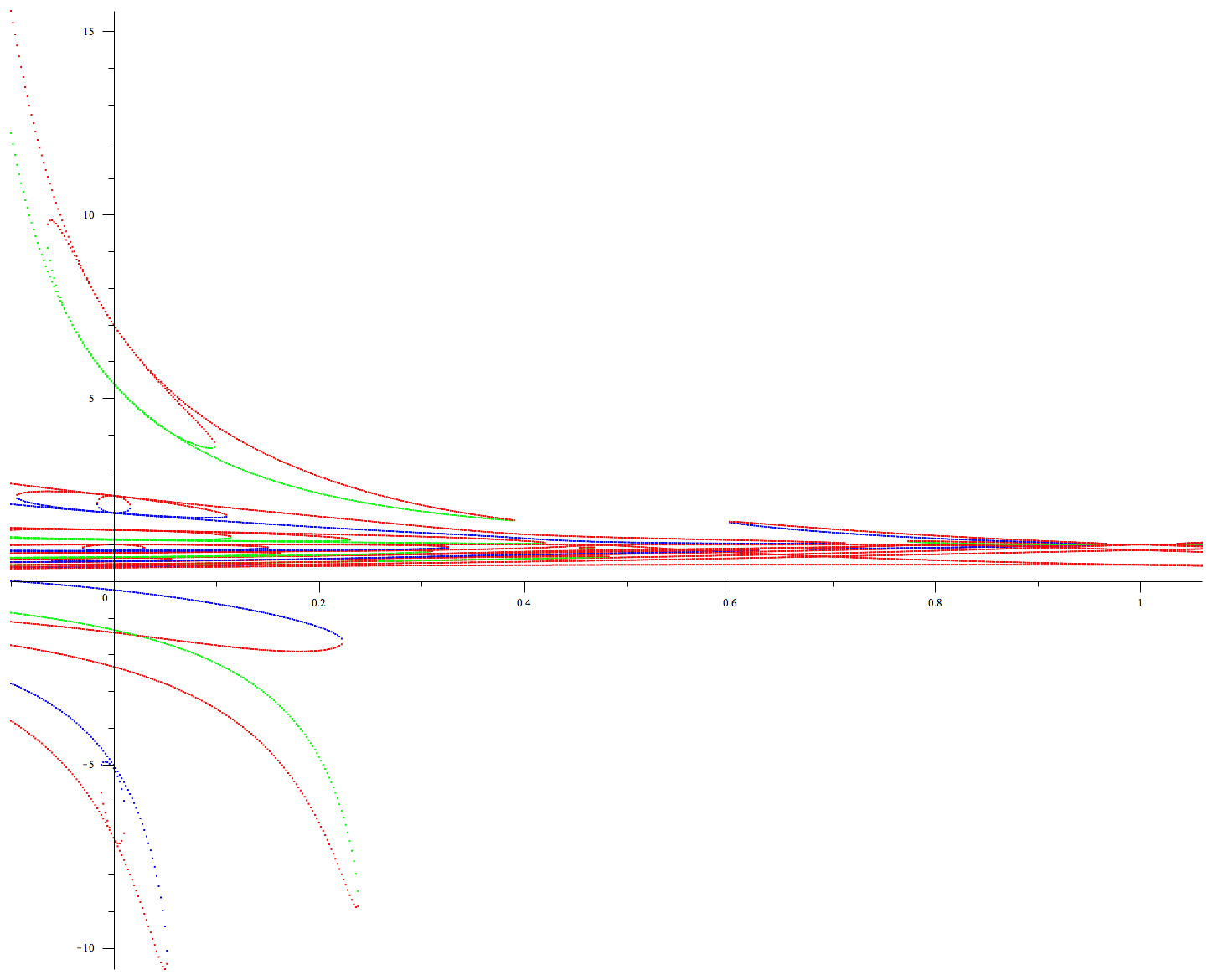} }}% 
        \centering
   % \subfloat  [ ]
    %{{\includegraphics[width=7cm,  height=6cm]{pics/N7p3.png} } }%
    ~~\\
     % \subfloat  [ ] 
    %{{\includegraphics[width=5cm]{pics/N7p5nberofcpts.png} } }%
    %\qquad    \quad \quad \, 
      \subfloat  [ ] 
    {{\includegraphics[scale=0.5]{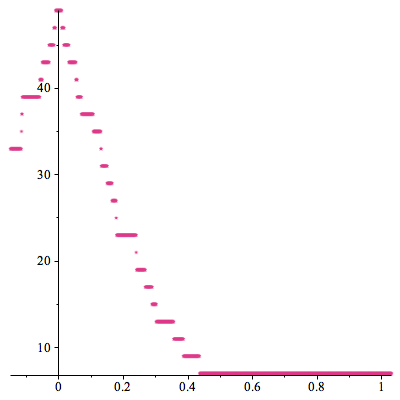} } }%
    \caption{ (a): The evolution of the critical points for $p=5$ and $n=7$. In red are the saddle points (anti-crossings). In green, are the maxima, and in blue, the minima.  Euler characteristic  is constant  $\chi=-7$ for all $b$ in $\mathbb R$. (b): 
    The number of critical points. This number is minimized at $b =1$ and maximized at $b =0$.}
    \label{dance}
  \end{figure}

%  \begin{sidewaysfigure}[ht]
%    \includegraphics[width=15cm]{pics/N7p5.png}
%    \caption{Caption in landscape to a figure in landscape.}
%    \label{fig:LandscapeFigure}
%\end{sidewaysfigure}

%   \begin{figure}[h]
%    \centering
%    \subfloat  []
%    {{\includegraphics[width=11cm]{pics/N7p5.png} }}% 
%    \caption{Evolution of the critical points for $p=5$ and $N=7$. Euler characteristic, in purple, is constant  $\chi=-7$.}
%    \label{charts}
%  \end{figure} 
%  

\section{Conclusion}  
In this paper we have presented a topological (qualitative) and geometrical (quantitative) description of the quantum  adiabatic evolution for finite dimensional Hamiltonians. The topological description, based on gradient flows and Morse homology, gives the global picture of the evolution: critical points and how they connect to each other. The differential geometry description uses Gau\ss-Bonnet to  zoom in, consistently with the topology, around the critical points to obtain the delay factors in the quantum adiabatic evolution.  This global-to-local procedure can potentially serve as the foundation for a practical tool for designing efficient Hamiltonians for adiabatic quantum computations.

~~\\
It is comforting,  remarkable and,  appealing, to see that the  mathematics we have employed fit naturally with the physics. It is our modest hope that our approach opens up new avenues for studying quantum speedup in adiabatic quantum computations.  

 \section*{Acknowledgements}
 We thank Dejan Slepcev of CMU for providing feedback on an earlier version of this paper, and members of the NASA QuAIL group (especially Eleanor Rieffel, Bryan O'Gorman, Davide Venturelli, and David Bell) for hosting us at USRA. 
 %{\color{blue}Julie?}

\bibliography{c}

\def\cprime{$'$} \def\cprime{$'$}
\begin{thebibliography}{10}

\bibitem{1366223}
D.~Aharonov, W.~van Dam, J.~Kempe, Z.~Landau, S.~Lloyd, and O.~Regev.
\newblock Adiabatic quantum computation is equivalent to standard quantum
  computation.
\newblock In {\em 45th Annual IEEE Symposium on Foundations of Computer
  Science}, pages 42--51, Oct 2004.

\bibitem{MR947141}
V.~I. Arnol\cprime~d.
\newblock {\em Geometrical methods in the theory of ordinary differential
  equations}, volume 250 of {\em Grundlehren der Mathematischen
  Wissenschaften}.
\newblock Springer-Verlag, New York, second edition, 1988.

\bibitem{Avron1999}
J.~E. Avron and A.~Elgart.
\newblock Adiabatic theorem without a gap condition.
\newblock {\em Communications in Mathematical Physics}, 203(2):445--463, 1999.

\bibitem{Hurtubise}
A.~Banyaga and D.~Hurtubise.
\newblock {\em Lectures on {M}orse homology}, volume~29 of {\em Kluwer Texts in
  the Mathematical Sciences}.
\newblock Kluwer Academic Publishers Group, Dordrecht, 2004.

\bibitem{Born1928}
M.~Born and V.~Fock.
\newblock Beweis des adiabatensatzes.
\newblock {\em Zeitschrift f{\"u}r Physik}, 51(3):165--180, 1928.

\bibitem{Bravyi:2008:CSL:2011772.2011773}
S.~Bravyi, D.~P. Divincenzo, R.~Oliveira, and B.~M. Terhal.
\newblock The complexity of stoquastic local hamiltonian problems.
\newblock {\em Quantum Info. Comput.}, 8(5):361--385, May 2008.

\bibitem{PhysRevA.65.012322}
A.~M. Childs, E.~Farhi, and J.~Preskill.
\newblock Robustness of adiabatic quantum computation.
\newblock {\em Phys. Rev. A}, 65:012322, Dec 2001.

\bibitem{0821816888}
C.~Conley.
\newblock {\em Isolated Invariant Sets and the Morse Index (Conference Board of
  the Mathematical Sciences Series No. 38)}.
\newblock American Mathematical Society, 1978.

\bibitem{doCarmo}
M.~P. do~Carmo.
\newblock {\em Differential geometry of curves and surfaces}.
\newblock Prentice-Hall, Inc., Englewood Cliffs, N.J., 1976.

\bibitem{1903.01486}
R.~Dridi, H.~Alghassi, and S.~Tayur.
\newblock Enhancing the efficiency of adiabatic quantum computations.
\newblock {A}rXiv:1903.01486, 2019.

\bibitem{Farhi472}
E.~Farhi, J.~Goldstone, S.~Gutmann, J.~Lapan, A.~Lundgren, and D.~Preda.
\newblock A quantum adiabatic evolution algorithm applied to random instances
  of an np-complete problem.
\newblock {\em Science}, 292(5516):472--475, 2001.

\bibitem{Grover:1996:FQM:237814.237866}
L.~K. Grover.
\newblock A fast quantum mechanical algorithm for database search.
\newblock In {\em Proceedings of the Twenty-eighth Annual ACM Symposium on
  Theory of Computing}, STOC '96, pages 212--219, New York, NY, USA, 1996.

\bibitem{MR1027662}
G.~A. Hagedorn.
\newblock Adiabatic expansions near eigenvalue crossings.
\newblock {\em Ann. Physics}, 196(2):278--295, 1989.

\bibitem{PhysRevA.74.052322}
S.~P. Jordan, E.~Farhi, and P.~W. Shor.
\newblock Error-correcting codes for adiabatic quantum computation.
\newblock {\em Phys. Rev. A}, 74:052322, Nov 2006.

\bibitem{Kato}
T.~Kato.
\newblock On the adiabatic theorem of quantum mechanics.
\newblock {\em Journal of the Physical Society of Japan}, 5(6):435--439, 1950.

\bibitem{Kitaev:2002:CQC:863284}
A.~Y. Kitaev, A.~H. Shen, and M.~N. Vyalyi.
\newblock {\em Classical and Quantum Computation}.
\newblock American Mathematical Society, Boston, MA, USA, 2002.

\bibitem{MR1873233}
Y.~Matsumoto.
\newblock {\em An introduction to {M}orse theory}, volume 208 of {\em
  Translations of Mathematical Monographs}.
\newblock American Mathematical Society, Providence, RI, 2002.

\bibitem{PhysRevLett.99.070502}
A.~Mizel, D.~A. Lidar, and M.~Mitchell.
\newblock Simple proof of equivalence between adiabatic quantum computation and
  the circuit model.
\newblock {\em Phys. Rev. Lett.}, 99:070502, Aug 2007.

\bibitem{Nielsen:2011:QCQ:1972505}
M.~A. Nielsen and I.~L. Chuang.
\newblock {\em Quantum Computation and Quantum Information: 10th Anniversary
  Edition}.
\newblock Cambridge University Press, New York, NY, USA, 10th edition, 2011.

\bibitem{10.3389/fict.2017.00002}
H.~Nishimori and K.~Takada.
\newblock Exponential enhancement of the efficiency of quantum annealing by
  non-stoquastic hamiltonians.
\newblock {\em Frontiers in ICT}, 4:2, 2017.

\bibitem{cerf}
J.~Roland and N.~J. Cerf.
\newblock Quantum search by local adiabatic evolution.
\newblock {\em Phys. Rev. A}, 65:042308, Mar 2002.

\bibitem{phdthesis}
T.~O. Rot.
\newblock {\em Morse Conley Floer Homology}.
\newblock PhD thesis, Vrije Universiteit Amsterdam, 2014.

\bibitem{Shor:1997:PAP:264393.264406}
P.~W. Shor.
\newblock Polynomial-time algorithms for prime factorization and discrete
  logarithms on a quantum computer.
\newblock {\em SIAM J. Comput.}, 26(5):1484--1509, Oct. 1997.

\bibitem{Suzuki2013}
S.~Suzuki, J.~ichi Inoue, and B.~K. Chakrabarti.
\newblock {\em Quantum Ising Phases and Transitions in Transverse Ising
  Models}.
\newblock Springer Berlin Heidelberg, 2013.

\bibitem{vazirani}
W.~van Dam, M.~Mosca, and U.~Vazirani.
\newblock How powerful is adiabatic quantum computation?
\newblock In {\em Proceedings 2001 IEEE International Conference on Cluster
  Computing}, pages 279--287, Oct 2001.

\bibitem{vonNeumann1993}
J.~von Neumann and E.~P. Wigner.
\newblock {\em {\"U}ber das Verhalten von Eigenwerten bei adiabatischen
  Prozessen}, pages 294--297.
\newblock Springer Berlin Heidelberg, Berlin, Heidelberg, 1993.

\bibitem{witten1982}
E.~Witten.
\newblock Supersymmetry and morse theory.
\newblock {\em J. Differential Geom.}, 17(4):661--692, 1982.

\end{thebibliography}

\end{document}